\documentclass[sigconf, 10pt]{acmart}
\AtBeginDocument{%
  }

\setcopyright{acmlicensed}
\copyrightyear{2026}
\acmYear{2026}
\setcopyright{cc}
\setcctype{by}
\acmConference[MobiCom '26]{The 32nd Annual International Conference on Mobile Computing and Networking}{October 26--30, 2026}{Austin, TX, USA}
\acmBooktitle{The 32nd Annual International Conference on Mobile Computing and Networking (MobiCom '26), October 26--30, 2026, Austin, TX, USA}
\acmDOI{10.1145/3795866.3844486}
\acmISBN{979-8-4007-2505-0/2026/10}

\usepackage{tikz}
\usepackage[utf8]{inputenc}
\usepackage[normalem]{ulem}
\usepackage{subcaption}
\newcommand{\name}{{mmFHE}}

\usepackage{xcolor}

\definecolor{revcolor}{RGB}{200, 0, 0}
\newif\ifshowrevisions
\showrevisionsfalse
\ifshowrevisions
  \newcommand{\rev}[1]{{\color{revcolor}#1}}
  \newcommand{\revdel}[1]{{\color{revcolor}\sout{#1}}}
  \newenvironment{revblock}{\par\color{revcolor}}{\par}
\else
  \newcommand{\rev}[1]{#1}
  \newcommand{\revdel}[1]{}
  \newenvironment{revblock}{}{}
\fi

\begin{document}

\title{\name{}: mmWave Sensing with End-to-End Fully Homomorphic Encryption}

% \affiliation{%

\author{Tanvir Ahmed}
\affiliation{%
  \institution{Cornell Tech}
  % \city{New York}
  % \state{NY}
  \country{{ }}
}
\email{tanvir@infosci.cornell.edu}
\orcid{0000-0002-9468-5033}

\author{Yixuan Gao}
\affiliation{%
  \institution{Cornell Tech}
  % \city{New York}
  % \state{NY}
  \country{{ }}
}
\email{yixuan@cs.cornell.edu}
\orcid{0000-0003-1778-3104}

\author{Adnan Armouti}
\affiliation{%
  \institution{Cornell Tech}
  % \city{New York}
  % \state{NY}
  \country{{ }}
}
\email{armouti@cs.cornell.edu}
% \orcid{0000-0002-9468-5033}

\author{Rajalakshmi Nandakumar}
\affiliation{%
  \institution{Cornell Tech}
  % \city{New York}
  % \state{NY}
  \country{{ }}
}
\email{rajalakshmi.nandakumar@cornell.edu}
\orcid{0000-0002-1601-148X}

% \author{Tanvir Ahmed, Yixuan Gao, Adnan Armouti, Rajalakshmi Nandakumar}
% % \author{Yixuan Gao}
% % \author{Adnan Armouti}
% % \author{Rajalakshmi Nandakumar}

% \affiliation{%
%   \institution{Cornell Tech}
%   \state{New York}

% % \email{tanvir@infosci.cornell.edu}
% % \orcid{0000-0002-9468-5033}

% % \email{yixuan@cs.cornell.edu}
% % \orcid{0000-0003-1778-3104}

% % \email{armouti@cs.cornell.edu}
% % % \orcid{0000-0002-9468-5033}

% % \email{rajalakshmi.nandakumar@cornell.edu}
% % \orcid{0000-0002-1601-148X}
% \email{\{tanvir@infosci.,yixuan@cs.,armouti@cs.,rajalakshmi.nandakumar@\}cornell.edu}

\renewcommand{\shortauthors}{Ahmed et al.}

\begin{abstract}
%-------------------------------------------------------------------------------

%-------------------------------------------------------------------------------
We present \name{}, the first system that \rev{executes the entire cloud-side mmWave sensing pipeline including the DSP and ML inference under fully homomorphic encryption (FHE)}. \name{} encrypts range profiles on \rev{an} edge device \rev{after lightweight plaintext preprocessing} and executes the entire mmWave signal-processing and ML inference pipeline homomorphically on a \rev{semi-honest} cloud that operates exclusively on ciphertexts. At the core of \name{} is a library of seven composable, data-oblivious FHE kernels that replace standard DSP routines with fixed arithmetic circuits for different application-specific pipelines. \revdel{These kernels can be flexibly composed into different application-specific pipelines.} We demonstrate this approach on two representative tasks: vital-sign monitoring and gesture recognition. We formally prove two cryptographic guarantees for any pipeline assembled from this library: \textit{input privacy} \revdel{, the cloud learns nothing about the sensor data;} and \textit{data obliviousness}\revdel{, the execution trace is identical on the cloud regardless of the data being processed}. These guarantees effectively neutralize various supervised and unsupervised privacy attacks on raw data, including re-identification and data-dependent privacy leakage. Evaluation on three public radar datasets shows that encryption introduces negligible error\revdel{: HR/RR MAE ${<}\,10^{-3}$\,bpm} versus \rev{the} plaintext \rev{pipeline}, with 84.5\% gesture accuracy (vs.\ 84.7\%).\revdel{With our batched ciphertext packing, }\rev{ End-to-end cloud GPU latency is 1.21\,s per 10\,s vital-sign window and 5.76\,s per 3\,s gesture window. These results establish the initial feasibility of end-to-end mmWave sensing under FHE on commodity hardware.}

\end{abstract}


% \keywords{Do, Not, Use, This, Code, Put, the, Correct, Terms, for,



\maketitle

\section{Introduction}
\label{sec:intro}

Millimeter-wave (mmWave) radar has emerged as a new sensing modality for mobile health and ubiquitous computing, enabling continuous and unobtrusive monitoring of physiological and behavioral metrics such as respiration~\cite{adib2015smart}, heart rate~\cite{alizadeh2019remote}, gait~\cite{meng2020gait}, sleep stages~\cite{zhao2017learning}, and fall detection~\cite{jin2020mmfall}. Compared \rev{to} cameras, radars operate unobtrusively in the background, provide privacy and avoid capturing identifying visual information, properties that make it well suited for continuous in-home health monitoring.

However, as the complexity of mmWave applications increases, so does the required computation. Edge devices often need to offload heavy computation (data processing, ML model inference) to cloud services. But offloading raw sensor data to a server introduces privacy risks on multiple fronts. First, transferring raw data to a remote server is inherently risky as not all cloud servers can be trusted. Second, raw radar data encodes far more than the target task requires, and a \rev{semi-honest (honest-but-curious)} server can extract unintended information from it. For instance, a server running a gesture recognition system receives raw signals that also contain the user's heartbeat, respiration, and biometric identity. Third, even without accessing data content, a server can exploit data-dependent execution patterns in the processing pipeline to infer the input~\cite{ohrimenko2016oblivious}. The same risk applies to sensing pipelines, where data-dependent DSP routines expose scene content through execution behavior.

\rev{A natural} solution \rev{to these privacy risks} is to perform all computations on the device without sending data to the cloud. However, that significantly limits the applications that can be achieved. Recent work, RPRS~\cite{wu2025rprs}, has explored using fully homomorphic encryption (FHE) as a cryptographic tool to partially address this problem: the entire upstream DSP chain executes in plaintext on the device, and only the final representation is encrypted before being sent to the cloud for neural-network inference. This protects the model weights, but not the user: all intermediate representations remain visible to the local server, and data-dependent DSP on unencrypted inputs can still leak information through timing side channels.

We present \name{}, the first system to \rev{execute the entire cloud-side mmWave DSP and ML inference pipeline under FHE}. Raw range profiles are encrypted on the edge client \rev{after lightweight plaintext preprocessing (clutter removal and normalization)} using the Cheon-Kim-Kim-Song (CKKS) FHE scheme; the ciphertexts are sent to a \rev{semi-honest} cloud that executes the full pipeline and returns an encrypted result; the client decrypts only the authorized outputs, heart rate (HR) / respiration rate (RR) in BPM, or a gesture logit vector. The cloud never observes any plaintext value at any stage. \name{} provides two formal privacy guarantees \rev{that follow by construction from CKKS and its fixed-circuit design;} input privacy: the cloud learns nothing about the content of the encrypted sensor stream (Theorem~\ref{thm:input_privacy}) and data obliviousness: the execution trace is identical regardless of what is being sensed (\rev{Theorem}~\ref{thm:data_oblivious}).

The main challenge \rev{of designing \name{}} is that FHE frameworks support only three operations on encrypted data: addition, multiplication, and rotation \rev{with a fixed depth budget}. Typical mmWave sensing pipelines, however, rely on standard DSP routines that are incompatible with these operations: FFT operation chains many sequential multiplications, CFAR detection branches on adaptive thresholds, target detection and phase extraction rely on argmax, vital-sign recovery requires $\arctan$, ML inference depends on non-linear activations such as ReLU. For mmWave sensing, \rev{to our knowledge,} none of these has a direct FHE equivalent.
 
To address this, we \textit{redesign} each stage of mmWave sensing as a fixed arithmetic circuit using only FHE-compatible operations. We replace FFTs with plaintext matrix-vector products that achieve the same computation at a fraction of the multiplicative cost, replace CFAR and argmax with a branch-free soft power attention mechanism that processes every bin uniformly, and approximate $\arctan$ via a fixed-order Taylor-\rev{series} expansion. The result is a library of seven composable FHE-compatible DSP kernels: energy integration, soft power attention, DFT, soft I/Q extraction of highest-energy bin, FIR filtering, notch masking, and differential phase extraction. We verify that each kernel's output matches its plaintext counterpart, and provide formal proofs of input privacy (Theorem~\ref{thm:input_privacy}) and data obliviousness (\rev{Theorem}~\ref{thm:data_oblivious}) for the end-to-end system.
 
We summarize the contributions as follows:

\begin{itemize}

    \item \textbf{Encrypted radar \rev{sensing} pipeline.} \name{} is the first system to encrypt the entire mmWave DSP and ML inference pipeline, from raw range profiles through detection, feature extraction, and classification, on a \rev{semi-honest (honest-but-curious)} cloud. No intermediate representation is ever exposed in plaintext, and two formal guarantees are provided: input privacy (Theorem~\ref{thm:input_privacy}) and data obliviousness (\rev{Theorem}~\ref{thm:data_oblivious}).

    \item \textbf{FHE-friendly DSP kernel library.} We design seven \rev{data-oblivious} composable kernels: energy integration, soft power attention, DFT, soft I/Q extraction of highest-energy bin, FIR filtering, notch masking, and differential phase extraction; each fitting within a shared multiplicative depth budget without bootstrapping.

    \item \textbf{Experimental validation.} We evaluate on three public radar datasets (270 vital-sign recordings, 600 gesture trials). The encrypted pipeline achieves HR/RR MAE ${<}\,10^{-3}$\,bpm \rev{versus the plaintext pipeline,} with per-kernel encryption noise below $10^{-5}$, and 84.5\% gesture accuracy (vs.\ 84.7\% plaintext, 99.8\% prediction agreement). \rev{With our batched ciphertext packing strategy, end-to-end cloud latency on a single GPU is 1.21\,s for a 10\,s vital-signs window (faster than the sensor data rate) and 5.76\,s for a 3\,s gesture window, with 22\,MB of uplink and 0.84\,s of client-side encryption per vital-signs window on a Raspberry Pi 4.}
\end{itemize}
\section{Motivation}
\label{sec:motivation}

Consider a mmWave radar deployed for sleep stage classification. The inference pipeline is computationally heavy, so raw sensor data must be offloaded to the cloud. But the cloud now receives far more than what the task requires: the same raw stream encodes breathing patterns, heart rate, and biometric motion signatures that identify the individual, and nothing prevents the cloud from extracting all of it. We demonstrate this concretely on two public datasets, neither of which was collected for the inferences we show are possible.

\paragraph{User re-identification from gesture data.} On a 60\,GHz gesture dataset~\cite{zenodo_gesture} (12 users, 5 gesture types), a random-forest classifier trained on 186 features from the raw IQ stream can achieve \textbf{99.9\%} user re-identification accuracy within a session (temporal 70/30 split) and \textbf{78.6\%} across physically different gesture types (train on \{1,2,3\}, test on \{6,7\}). Thus, the user identity is encoded in the same raw signal features that encodes the gestures.

\paragraph{Unsupervised linkability.} A \rev{semi-honest} cloud may not have any labeled data to train, unlike the previous case. But the labels may not be necessary to make unauthorized biometric inferences: computing pairwise cosine distances over standardized features, without training any classifier, suffices to link recordings to the same person. On a \rev{60}\,GHz vital-sign dataset~\cite{children_dataset} of 50 children, this \rev{training-free} attack achieves AUC\,=\,0.981 (Fig.~\ref{fig:motivation}\rev{(a)}). On the gesture dataset, AUC\,=\,0.876. UMAP projections confirm that per-user clusters emerge without supervision (Fig.~\ref{fig:motivation}\rev{(b)}).

% \paragraph{Side-channel leakage from data-dependent DSP} Standard CFAR detection branches on whether each range bin exceeds an adaptive threshold; its execution time therefore varies with the scene. Across 50 children, CFAR timing correlates with detection count at $r{=}+0.94$ ($p < 10^{-23}$, 2000 trials per subject). Thus, the cloud operator can infer target count and approximate range from timing alone, acquiring unauthorized environment information from the raw data.

\paragraph{Execution trace leakage from data-dependent algorithms} Beyond data content, the execution trace of a sensing pipeline can itself reveal information without accessing any data value. Ohrimenko et al.~\cite{ohrimenko2016oblivious} show that ML algorithms leak sensitive information through their execution patterns: the same program runs differently depending on what data it processes, and anyone observing how long it takes or how much memory it accesses can infer the input. In a radar pipeline, a routine that processes more data when more people are in the room, or runs longer when a target is moving, exposes scene information through timing alone, without ever seeing the encrypted data. This is the leakage channel formalized in Proposition~\ref{prop:side_channel_leakage}.

\begin{figure}[t]
  \centering
  \includegraphics[width=1\columnwidth]{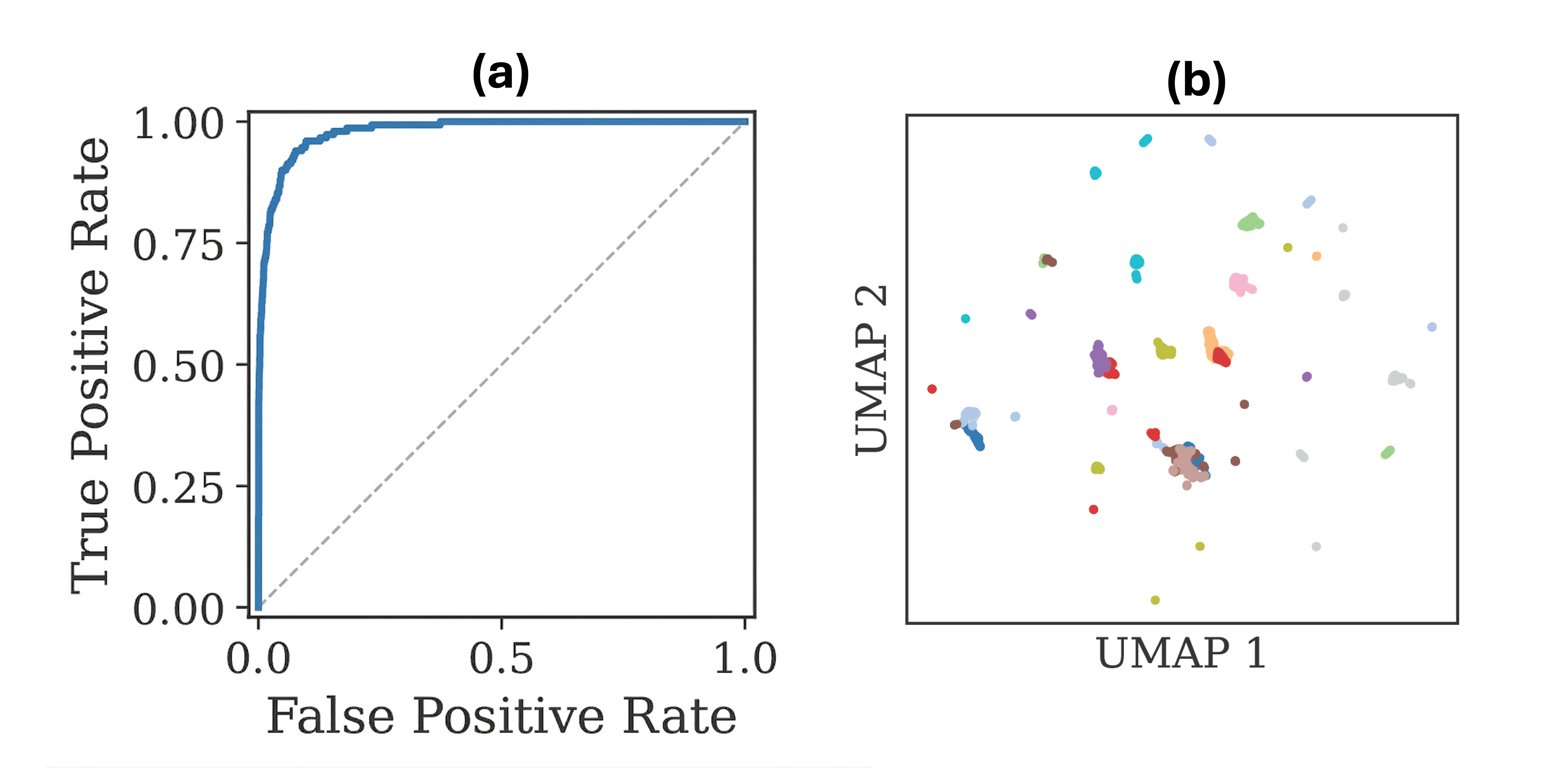}
  \caption{Unsupervised privacy leakage from raw mmWave data. \textbf{a)} Linkability ROC on 50 \rev{users~\cite{children_dataset}} (AUC\,=\,0.981). \textbf{b)} UMAP of gesture features (12 users) \rev{on \cite{zenodo_gesture} dataset}; per-user clusters form without labels. }
  \label{fig:motivation}
\end{figure}

These risks share the same root cause: processing raw data encodes more information than the intended application requires, and data-dependent DSP amplifies the privacy risks through observable execution behavior. We formalize these observations into the following propositions:

\begin{proposition}[Biometric Leakage \rev{from Plaintext} Radar Streams]
\label{prop:biometric_inseparability}
Let $\{\mathbf{z}[r, t]\}_{t=1}^{F}$ be a plaintext FMCW range-profile stream with carrier wavelength $\lambda$, slow-time sampling rate $f_s \geq 2 f_{\mu}$ (where $f_{\mu}$ is the highest micro-motion frequency of interest), and observation duration $F/f_s \geq T_{\min}$. Any protocol that grants a semi-honest server plaintext access to this stream, even for a limited task such as gesture recognition or presence detection, necessarily leaks sufficient information to reconstruct the target's biometric information (that is encoded in frequencies up to $f_{\mu}$) such as heartbeat, respiration, gait.
\end{proposition}

% \begin{proposition}[Inseparability Property and Biometric Leakage from Raw Radar Data]
% \label{thm:biometric_inseparability}
% Raw coherent radar streams encode sub-wavelength micro-motion signatures---heartbeat, respiration, gait---that uniquely identify observed individuals. Any protocol that grants a semi-honest server plaintext access to these streams, even for a limited task such as gesture recognition or presence detection, necessarily exposes sufficient information to extract biometric identifiers and vital signs of every target in the scene.
% \end{proposition}

\begin{proposition}[Side-Channel Leakage from Data-Dependent DSP]
\label{prop:side_channel_leakage}
Standard radar DSP routines (e.g., CFAR detection, Kalman filtering) contain data-dependent control flow that branches on input values. A semi-honest server can distinguish radar scenes from the observed execution trace alone, without accessing any data value.
\end{proposition}

\name{} addresses both by encrypting raw radar data before it leaves the client and executing the entire \rev{sensing} pipeline as fixed arithmetic circuits on ciphertexts.

\section{\name{}}
\label{sec:system}
 
\name{} enables privacy preserving mmWave sensing with end-to-end fully homomorphic encryption (FHE). \revdel{It starts by encrypting raw mmWave complex range profiles on the edge and sending the encrypted data to cloud, which executes the entire DSP and inference pipeline on encrypted data.} Figure~\ref{fig:arch} illustrates the end-to-end architecture: the client encrypts the raw complex range profile radar output; the cloud chains a library of FHE-compatible kernels over encrypted radar data and sends back the encrypted result; the client then decrypts only the authorized scalar outputs. The cloud never observes any unencrypted intermediate results at any stage. 
 
\label{subsec:architecture}
% \begin{figure*}[t]
%     \centering
%     \resizebox{0.95\textwidth}{!}{\input{figures/architecture_diagram.tex}}
%     \caption{\name{} architecture. The client performs lightweight
%     pre-processing and encryption; the cloud executes the full
%     encrypted DSP and inference pipeline using composable kernels; the
%     client decrypts only authorized scalar outputs. Two application
%     compositions are shown: vital signs (left path) and ML inference
%     (right path). The secret key $sk$ is shared only between the
%     trusted client endpoints.}
%     \label{fig:arch}
% \end{figure*}

\begin{figure*}[t]
    \centering
    \includegraphics[width=0.95\linewidth]{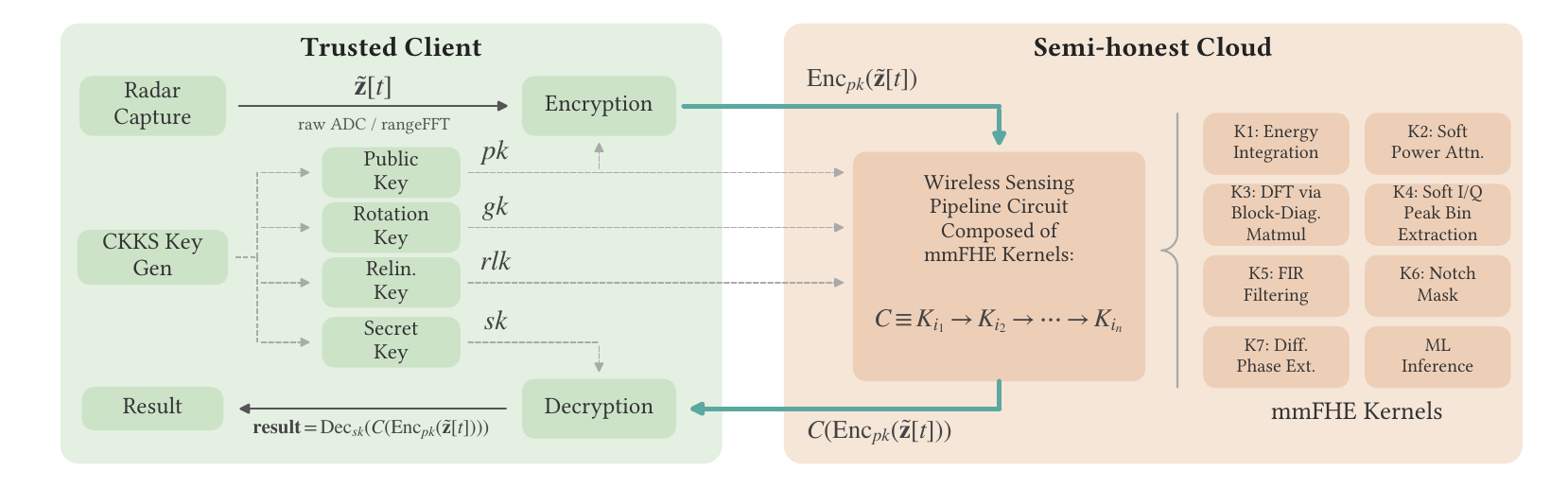}
    \caption{\name{} architecture. The trusted edge client encrypts raw range-FFT profiles under CKKS and sends ciphertexts to a \rev{semi-honest} cloud, which executes the entire DSP and ML inference pipeline using composable \rev{data-oblivious} FHE kernels (K1--K7) over encrypted data. The cloud never observes any plaintext value. The client decrypts only authorized outputs (HR/RR BPM or waveform, or classification logits). Legend: black arrows = unencrypted plaintext, teal arrows = encrypted ciphertext, dashed arrows = keys.}
    \label{fig:arch}
\end{figure*}

\subsection{Threat Model}
\label{subsec:threat_model}
 
\noindent\textbf{Adversary.} We \rev{assume a} semi-honest (honest-but-curious) \rev{cloud adversary}: the cloud faithfully executes a provided protocol but may analyze all observed inputs, intermediate values, and execution timing. Two independent leakage channels arise: (1) content leakage from plaintext data access, and (2) side-channel leakage from data-dependent execution traces. The \rev{cloud} adversary is computationally bounded and cannot break the Ring \rev{Learning with Errors (R}LWE) assumption (\revdel{a  cryptographic hardness assumption where }i.e. it is computationally infeasible to distinguish noisy, structured polynomial equations from truly random ones).
 
\noindent\textbf{Trust boundary.} The client (edge radar device) is trusted: it generates cryptographic keys, holds the secret key, enforces output disclosure policies, and encrypts radar output into ciphertexts. The cloud receives only ciphertexts and public evaluation keys.
 
\noindent\textbf{Security goals.} \name{} \rev{must} provide two formal guarantees. (1)~\textit{Input privacy}: the cloud learns nothing about the content of the encrypted sensor stream (Theorem~\ref{thm:input_privacy}: input is indistinguishable from noise) and (2)~\textit{Data obliviousness}: the execution trace is identical regardless of what is being sensed (\rev{Theorem}~\ref{thm:data_oblivious}). Proofs appear in \S\ref{sec:security}.
 
\noindent\textbf{Assumptions and limitations.} \revdel{\name{} assumes the client device is trusted and physically secure; \rev{side-channel} attacks on the client hardware are out of scope. }
\rev{(1)~\emph{Malicious computation:}} The system does not protect against a malicious cloud that deviates from the protocol (e.g., not running the provided \rev{circuit} faithfully or returning values not computed on the original client's data), only semi-honest adversaries are considered. \rev{(2)~\emph{Metadata:} } \rev{network-level observables} such as \rev{ciphertext counts, sizes, and the client's network identity are visible to the adversary.} \rev{(3)~\emph{Session timing:}} The timing and frequency of sensing sessions are not protected. \rev{(4)~\emph{Authorized outputs:} whatever the trusted client decrypts is disclosed by definition; e.g., a decrypted HR/RR waveform may expose the same signal that enables the linkage attack in Figure~\ref{fig:motivation}.}

\subsection{Homomorphic Encryption Primer}
\label{subsec:he_primer}
 
\revdel{In \name{}, the client encrypts its input, sends the ciphertext to the server, and the server \rev{computes and} returns an encrypted result that only the client can decrypt,}\name{} is based on the CKKS FHE scheme~\cite{cheon2017homomorphic}. CKKS is designed for approximate arithmetic on real and complex \rev{non-integer} numbers, suitable for mmWave sensing. CKKS security relies on the RLWE problem~\cite{lyubashevsky2010ideal}, \rev{and is} believed to resist both classical and quantum attacks. We follow the HE Standard~\cite{albrecht2022homomorphic} to achieve $\geq 128$-bit security.

\noindent\textbf{Key Generation.} The client runs $\mathsf{KeyGen}$ once to produce four artifacts: a secret key $sk$, a public encryption key $pk$, a relinearization key $rlk$, and a set of Galois keys $gk$. The secret key $sk$ is used for decryption and never leaves the client. The public key $pk$ encrypts plaintext vectors into ciphertexts. The relinearization key $rlk$ is required after every ciphertext-ciphertext multiplication to reduce the ciphertext back to two polynomials; without it, ciphertext size grows with each multiplication operation. Galois keys $gk$ enable slot rotations; one key is needed per distinct rotation amount $k$. The client sends $(pk, rlk, gk)$ to the cloud along with the encrypted data. The cloud uses these keys during homomorphic evaluation but cannot decrypt any result \rev{because it does not possess $sk$}.

\noindent\textbf{Encoding and Batching.} A single ciphertext packs up to $N/2$ real values into parallel slots, where $N$ is a CKKS parameter. All arithmetic applies slot-wise, so one ciphertext operation processes up to $N/2$ values simultaneously. \revdel{With $N{=}32768$, a single encrypted add or multiply acts on 16{,}384 values at once.} This single instruction, multiple data (SIMD) batching makes encrypted radar DSP practical: an entire range profile or Doppler spectrum fits in one ciphertext.

\noindent\textbf{Three Primitive Operations.} All encrypted computation is built from three operations. \textit{Addition:} $\mathsf{Enc}(\mathbf{a}) + \mathsf{Enc}(\mathbf{b}) = \mathsf{Enc}(\mathbf{a}+\mathbf{b})$, slot-wise, free in depth. \textit{Multiplication:} $\mathsf{Enc}(\mathbf{a}) \cdot \mathsf{Enc}(\mathbf{b}) = \mathsf{Enc}(\mathbf{a} \odot \mathbf{b})$, slot-wise, consumes one depth level. \textit{Rotation:} $\mathsf{Rot}(\mathsf{Enc}(\mathbf{a}), k)$ cyclically shifts the slot vector by $k$ positions, free in depth but requires $gk$. Non-linear functions such as argmax comparison, and $\arctan$ have no direct \rev{FHE} counterpart and must be approximated as polynomial circuits subject to a depth budget $L$. 
 
\noindent\textbf{Depth Budget.} Every CKKS ciphertext carries a finite multiplicative depth budget $L$. Each ciphertext-ciphertext multiplication consumes one level; after $L$ multiplications the ciphertext cannot be used further. Addition is free. The entire DSP pipeline must be expressed as a circuit with at most $L$ sequential multiplications. Deeper circuits require either a larger $N$ or bootstrapping, a costly operation that refreshes the depth budget. \name{} pipelines are designed to fit within their budgets without bootstrapping.

%%%%%%%%%%%%%%%%%%%%%%%%%%%%%%%%%%%%%%%%%%%%%%%%%%%%%%%%%%%%%%%%%%%%
%% Section: System Design
%%%%%%%%%%%%%%%%%%%%%%%%%%%%%%%%%%%%%%%%%%%%%%%%%%%%%%%%%%%%%%%%%%%%
\subsection{System Design}
\label{sec:system_design}
\label{sec:pipeline}
% The vital signs pipeline requires two lightweight decryptions 
% (target-bin index and
% HR/RR BPM values); the classification pipeline requires a single
% logit-vector decryption. All decryptions occur on the trusted client.

%-------------------------------------------------------------------
\subsubsection{Overview}

\revdel{\name{} prevents both biometric and side-channel leakage, and breaks the inseparability identified in Prop.~\ref{prop:biometric_inseparability} and \ref{prop:side_channel_leakage}, by encrypting outgoing radar streams on the client and executing the entire signal-processing and inference chain over ciphertexts on the cloud.} Our key contribution is a set of data-oblivious, FHE-compatible mmWave sensing kernels that enable the cloud to execute standard DSP-like algorithms without ever accessing plaintext data. Because the cloud never observes plaintext values, \name{} prevents the biometric leakage in Prop.~\ref{prop:biometric_inseparability}. Moreover, data obliviousness closes both the content leakage and side-channel leakage identified in \S\ref{sec:motivation}. These kernels can be composed into complete mmWave sensing pipelines: the cloud chains them over ciphertexts, and the trusted client decrypts only the authorized scalar outputs. \revdel{Figure~\ref{fig:arch} summarizes the system architecture, which consists of two main components:} \name{} has two main actors (Fig. ~\ref{fig:arch}).

\noindent (1) Trusted \textbf{Client} does the following:

      \noindent\textit{Pre-encryption stage:} after receiving the range FFT from the mmWave sensor, removes clutter and normalizes the data.\\
      \noindent\textit{Encryption stage:} generates CKKS keys $(pk, sk, rlk, gk)$. The secret key $sk$ never leaves the trusted client device(s). The client encrypts the normalized data using $pk$, packs the encrypted data, and sends the ciphertext \revdel{together} with $rlk, pk, gk$ to the cloud. \\
      \noindent\textit{Decryption and output stage:} after receiving encrypted results from the cloud, decrypts them using $sk$ and recovers only the authorized outputs, such as HR/RR waveforms, HR/RR BPM values, a target-bin index, or the logit vector for classification.

\noindent (2) Semi-honest \textbf{Cloud} receives the ciphertext and $rlk, pk, gk$ from the client. It selects and chains kernels from the \name{} library according to the target application. Each kernel is a fixed arithmetic circuit, and \revdel{some kernels require $gk$ to} operate on ciphertexts. The cloud returns the encrypted result to the client and never observes any plaintext value.

%-------------------------------------------------------------------
\subsubsection{Pre-Cloud Client-Side Algorithms}
\label{subsec:client_preprocessing}

\revdel{The pre-cloud client-side algorithms in \name{}} These depend on the type of radar stream available at the client: (i) raw I/Q data, (ii) range-FFT data, or (iii) highly processed lossy data. Since most commercial, non-research-grade radars directly provide a range-FFT-like stream, we assume type (ii) in this paper. However, as will become clear shortly, \name{} can also be integrated with type (i) radars.

\noindent\textbf{Preprocessing.}
We start from the range FFT (Eq.~\ref{eq:range_fft}) for each antenna and chirp across $R$ range bins, and apply clutter removal through mean subtraction:
$\tilde{\mathbf{z}}[t] = \mathbf{z}[t] -
\frac{1}{F}\sum_{l}\mathbf{z}[l] $. The resulting values are then scaled to a numerically stable range for FHE.
% frame-local normalization for vital signs
% (the scalar cancels in differential phase), batch-global for gestures. 

% The lowest range bins are zeroed to suppress near-field
% leakage. 
% This is the only FFT in the pipeline; all subsequent
% spectral operations are reformulated as matrix-vector products
% executable under CKKS.

\noindent\textbf{Encryption.}
 CKKS operates on real-valued slot vectors, \rev{so} complex radar data \rev{frame} $\tilde{\mathbf{z}}[t] = \mathbf{v}^{(\mathrm{re})}\rev{[t]}+j\,\mathbf{v}^{(\mathrm{im})}\rev{[t]}$ 
is split into separate real and imaginary ciphertexts: $\hat{\mathbf{v}}\rev{[t]}^{(\mathrm{re})} = \mathsf{Enc}_{pk}(\mathbf{v}\rev{[t]}^{(\mathrm{re})})$, $\hat{\mathbf{v}}\rev{[t]}^{(\mathrm{im})} = \mathsf{Enc}_{pk}(\mathbf{v\rev{[t]}}^{(\mathrm{im})})$; where \rev{the ciphertext pair $(\hat{\mathbf{v}}\rev{[t]}^{(\mathrm{re})}, \hat{\mathbf{v}}\rev{[t]}^{(\mathrm{im})})$ constitutes one} encrypted frame. 

\rev{\noindent\textbf{Packing.}}
The ciphertext packing \rev{(i.e., how frames are laid out across slots, and how many frames share one ciphertext)} can be simple \rev{per-frame} or \rev{batched}, depending on the downstream application. \\
\rev{\emph{Per-frame:} each frame occupies its own ciphertext pair. Within a frame, the layout follows the application. For} vital-signs, each frame is packed across $R$ range bins into the first $R$ slots. For Doppler pipelines (e.g., dynamic classification), to reduce the number of rotations required in the later stages, per-frame data from $A$ antennas $\times$ $R$ range bins $\times$ $D$ chirps \rev{can be} interleaved using $\text{slot}[a \cdot RD + r \cdot D + c]$, yielding $ARD$ active slots. This layout places each antenna-range group's $D$ chirps in contiguous slots, enabling the block-diagonal Doppler DFT to operate via slot rotations within each $D$-element block.
\rev{However, in this packing strategy each ciphertext uses only a small fraction of its slots (both $R,\,ARD \ll N/2$, yielding poor slot utilization) and the ciphertext count grows linearly with the window length. \\
\emph{Batched:} multiple frames are packed into one ciphertext pair, to maximize slot utilization. Because homomorphic operations are SIMD, packing $k$ frames per ciphertext divides the number of ciphertext operations for the element-wise stages by $k$, and reduces the ciphertext count from $O(F)$ to $O(F/k)$, cutting client encryption time and uplink volume by the same factor. In exchange, combining frames now costs slot rotations within a ciphertext instead of rotation-free additions across ciphertexts.}

%-------------------------------------------------------------------
\subsubsection{Cloud-Side Algorithms}
\label{subsec:cloud_processing}

The cloud receives ciphertexts and public parameters from the client, and evaluates application specific signal-processing kernels directly on encrypted data. In a conventional plaintext sensing pipeline, these operations include power computation, highest-energy-bin selection, phase extraction, CFAR, FIR filtering, FFT computation, frequency estimation, and related steps. A primer on conventional FMCW signal-processing algorithms is provided in Appendix~\ref{subsec:fmcw_primer}.

The main challenge is that many standard DSP algorithms are not directly compatible with FHE. In particular, FHE does not support data-dependent branching or operations that cannot be expressed as compositions of addition, multiplication, and rotation. As a result, a plaintext sensing pipeline cannot be ported to FHE by simply replacing plaintext values with ciphertexts.

To address this gap, \name{} provides a library of cloud-side kernels that together replicate the functionality of a plaintext DSP pipeline over ciphertexts. For operations that are naturally compatible with FHE, such as DFT and FIR filtering, \name{} redesigns them to execute efficiently under homomorphic constraints. For operations that are not directly compatible with FHE, such as highest-energy-bin selection, phase extraction, and CFAR, \name{} constructs FHE-friendly approximations that preserve the intended functionality. These kernels can be used individually or composed into complete sensing pipelines.

All \name{} kernels are data-oblivious: they execute the same sequence of additions, multiplications, and rotations regardless of the encrypted input. As a result, the cloud-side execution trace does not reveal information about the underlying radar scene. We describe these kernels below.

%- - - - - - - - - - - - - - - - - - - - - - - - - - - - - - - - - -
\paragraph{Kernel 1: Energy Integration}
\label{subsec:energy}
This kernel computes the per-bin signal power. For each range bin $r$, the cloud evaluates:
\begin{equation}
  E_r \;=\; \sum_{t=1}^{F}
  \Bigl(\mathrm{Re}(\tilde{z}_r[t])^2 + \mathrm{Im}(\tilde{z}_r[t])^2\Bigr).
  \label{eq:energy}
\end{equation}

This is the standard energy accumulation step used to identify bins with consistently strong reflections over time. Under FHE, the cloud computes $E_r$ on a fresh ciphertext, which consumes one multiplicative level, and then accumulates the result across frames using homomorphic additions, which are depth-free. The output is a single packed ciphertext containing the energies of all $R$ range bins.

%- - - - - - - - - - - - - - - - - - - - - - - - - - - - - - - - - -
\paragraph{Kernel 2: Soft Power Attention}
\label{subsec:soft_argmax}
This kernel replaces standard radar CFAR thresholding or hard argmax in \name{} with a branch-free polynomial attention mechanism that processes every bin uniformly. Given an energy profile $E_r$, \name{} computes the power-weighted statistics:
\begin{equation}
  w_r = E_r^{\gamma}, \label{eq:soft_power_weight}
  \qquad 
  num = \sum_{r=0}^{R-1} r\,w_r,
  \qquad 
  den = \sum_{r=0}^{R-1} w_r,
  % \label{eq:soft_argmax_stats}
\end{equation}
where the sharpening exponent $\gamma$ controls how strongly the weights concentrate on dominant bins. When $\gamma{=}1$, the weighted statistics reduce to an energy centroid; as $\gamma \to \infty$, the weight concentrates on the maximum-energy bin, approaching hard argmax. This enables a first-moment estimate of the target bin, $\hat{r} = num/den$.

This formulation is easily implemented in FHE because it avoids data-dependent branching. The weights $w_r$ are computed through sequential square operations, each consuming depth~1; larger $\gamma$ yields sharper selection at the cost of higher depth. In our experiments, $\gamma{=}2$ is sufficient for range-energy profiles, while $\gamma{=}4$ provides sharper selection when applied to higher-dimensional sensing maps such as range-angle maps, and range-Doppler maps.

%- - - - - - - - - - - - - - - - - - - - - - - - - - - - - - - - - -
% tmp_kernel_dft_block_diagonal_matmul.tex

\paragraph{Kernel 3: DFT (via Block-Diagonal Matmul)}
\label{subsec:doppler_dft}

The DFT is a core operation in radar sensing, used to transform temporally or spatially organized measurements into Doppler or angle-domain representations. In plaintext systems, this transform is typically implemented using FFT butterflies. Under FHE, however, each butterfly stage consumes one multiplicative depth level, making the standard FFT realization inefficient. To reduce this cost, \name{} instead expresses the windowed, shifted DFT as a plaintext matrix-vector multiplication, which decomposes into additions, multiplications by plaintext constants, and rotations while consuming only a single depth level. This matrix-vector product is implemented using the Halevi-Shoup diagonal method~\cite{halevi2014algorithms} with baby-step/giant-step (BSGS) scheduling, which groups and reuses rotations to substantially reduce the number of ciphertext rotations and required rotation keys in practice.

The windowed DFT kernel ($D \times D$) is:
\begin{equation}
  \mathbf{W}=W_{d,n} = w[n] \cdot e^{-j\frac{2\pi}{D}\sigma(d) \cdot n},
  \quad d,n \in [D]
  \label{eq:dft_kernel}
\end{equation}
where $w[n]$ is the Hanning window and $\sigma$ is the
\texttt{fftshift} permutation ($\sigma(d) = (d+D/2) \bmod D$).
Splitting $\mathbf{W}$ into real and imaginary parts, $\mathbf{C} = \mathbf{W}^{(\mathrm{re})}$ and $\mathbf{S} = \mathbf{W}^{(\mathrm{im})}$, and then tiling them block-diagonally $AR$ times to match the slot layout, gives:
\begin{equation}
  \hat{\mathbf{d}}^{(\mathrm{re})} =
    \tilde{\mathbf{C}} \cdot \hat{\mathbf{v}}^{(\mathrm{re})}
    - \tilde{\mathbf{S}} \cdot \hat{\mathbf{v}}^{(\mathrm{im})},
  \,
  \hat{\mathbf{d}}^{(\mathrm{im})} =
    \tilde{\mathbf{S}} \cdot \hat{\mathbf{v}}^{(\mathrm{re})}
    + \tilde{\mathbf{C}} \cdot \hat{\mathbf{v}}^{(\mathrm{im})},
  \label{eq:dft_re}
\end{equation}
where $\tilde{\mathbf{C}} = \mathbf{I}_{AR} \otimes \mathbf{C}$ and $\tilde{\mathbf{S}} = \mathbf{I}_{AR} \otimes \mathbf{S}$ are the block-diagonal replications. This makes the DFT practical under FHE while preserving the same linear transform.

%- - - - - - - - - - - - - - - - - - - - - - - - - - - - - - - - - -
\paragraph{Kernel 4: Soft I/Q Extraction of Highest-Energy Bin}
\label{subsec:phase}
In vital-sign sensing, phase is typically extracted from the range bin with the highest energy. Under FHE, however, hard bin selection is not possible because it requires data-dependent control flow. \name{} therefore replaces this step with a soft, branch-free alternative. For each frame $t$, the cloud computes a per-frame energy mask that concentrates on the highest-energy range bin and uses it to extract soft-weighted I/Q scalars:
\begin{align}
  m_r[t] &= [\mathrm{Re}(\tilde{\mathbf{z}}_r[t])^2 + \mathrm{Im}(\tilde{\mathbf{z}}_r[t])^2]^{P_\phi},
  \label{eq:phase_mask} \\
  I[t] &= \sum_{r=0}^{R-1} m_r[t] \cdot \mathrm{Re}(\tilde{\mathbf{z}}_r[t]),
  \, 
  Q[t] = \sum_{r=0}^{R-1} m_r[t] \cdot \mathrm{Im}(\tilde{\mathbf{z}}_r[t]).
  \label{eq:phase_iq}
\end{align}
As $P_\phi$ increases, the mask becomes more concentrated on the dominant range bin, making $(\mathbf{I}, \mathbf{Q})=(I[t], Q[t])$ approach the I/Q values of the highest-energy bin. The resulting encrypted outputs remain in ciphertext form and feed directly into subsequent kernels without intermediate decryption. In our experiments, $P_\phi=2-5$ is sufficient.

%- - - - - - - - - - - - - - - - - - - - - - - - - - - - - - - - - -
\paragraph{Kernel 5: FIR Filtering}
\label{subsec:fir}
FIR (finite impulse response) filtering is easily translated to FHE because it is purely feedforward: each output is a weighted sum over a fixed input window, with no data-dependent control flow. For each pass-band $b$, the filtered I/Q signals are expressed as:
\begin{equation}
  \mathbf{I}^{(b)}_{f} = \mathbf{T}^{(b)} \mathbf{I},
  \qquad
  \mathbf{Q}^{(b)}_{f} = \mathbf{T}^{(b)} \mathbf{Q},
  \label{eq:fir_iq}
\end{equation}
where $\mathbf{T}^{(b)}$ is the Toeplitz matrix formed from the plaintext filter taps of pass-band $b$. This requires only plaintext-ciphertext multiplication together with additions and rotations.

%- - - - - - - - - - - - - - - - - - - - - - - - - - - - - - - - - -
\paragraph{Kernel 6: Notch Masking}
\label{subsec:notch}
For static clutter removal in the range-Doppler domain under FHE, \name{} applies a fixed binary plaintext mask:
\begin{equation}
  m[d] = \begin{cases}
    0, & d \in [\lfloor D/2 \rfloor, \lceil D/2 \rceil) \\
    1, & \text{otherwise}
  \end{cases}
  \label{eq:notch_mask}
\end{equation}
This mask suppresses the near-zero Doppler region associated with static clutter while preserving all other bins. The mask is a public constant, so it can be realized as a single depth-free plaintext-ciphertext multiplication.

%- - - - - - - - - - - - - - - - - - - - - - - - - - - - - - - - - -
\paragraph{Kernel 7: Differential Phase Extraction (via Taylor Series Approximation)}
\label{subsec:taylor_arctan}
To extract phase under FHE, \name{} replaces $\arctan(Q/I)$ with a low-order Taylor approximation that operates directly on consecutive filtered samples. Given $y[t]$ and $x[t]$ defined as follows:
\begin{align}
  y[t] &= \mathbf{Q}_f[t]\,\mathbf{I}_f[t{-}1] - \mathbf{I}_f[t]\,\mathbf{Q}_f[t{-}1],
  \nonumber \\
  x[t] &= \mathbf{I}_f[t]\,\mathbf{I}_f[t{-}1] + \mathbf{Q}_f[t]\,\mathbf{Q}_f[t{-}1],
  \nonumber \\
  \Delta\phi[t] &\approx y[t]\,x[t]^2 - \tfrac{1}{3}\,y[t]^3,
  \label{eq:taylor_arctan}
\end{align}
\name{} uses a third-order Taylor expansion of $\arctan(y/x)$, which requires three ciphertext-ciphertext multiplicative levels. This approximation avoids division and inverse trigonometric evaluation, both of which are expensive under FHE.

For low-SNR scenarios where $|\Delta\phi| \ll 1$, a first-order approximation $\Delta\phi \approx y$ is often sufficient, trading accuracy for a shallower circuit.

%- - - - - - - - - - - - - - - - - - - - - - - - - - - - - - - - - -
\paragraph{Encrypted ML Inference}
\label{subsec:mlp}
Following standard practice for encrypted neural inference~\cite{wu2025rprs,kim2022secure}, \name{} uses a fully connected (FC) network with square activations, $\sigma(x) = x^2$. The trained weights $\{\mathbf{W}^{(\ell)}, \mathbf{b}^{(\ell)}\}$ are provided to the cloud as public plaintext parameters. For each layer $\ell = 1, \ldots, L_{FC}$:
\begin{equation}
  \mathbf{h}^{(\ell)} =
  \begin{cases}
    \bigl(\mathbf{W}^{(\ell)} \mathbf{h}^{(\ell-1)} +
    \mathbf{b}^{(\ell)}\bigr)^{\circ 2}
    & \ell < L_{FC},\\[4pt]
    \mathbf{W}^{(L_{FC})} \mathbf{h}^{(L_{FC}-1)} + \mathbf{b}^{(L_{FC})}
    & \ell = L_{FC},
  \end{cases}
  \label{eq:mlp_forward}
\end{equation}
where $(\cdot)^{\circ 2}$ denotes element-wise squaring.

\noindent\textbf{Weight folding.}
The soft-power normalization factor $s^\gamma / F$ and the column permutation that reorders the encrypted slot layout to match the training feature order are both folded into $\mathbf{W}^{(1)}$ at setup time. These are plaintext operations on the weight matrix and therefore incur zero depth cost. Here, $s = R \cdot A \cdot (\sum_n w[n])^2$ denotes the expected spectral energy scale. {Depth consumed: $2L_{FC} - 1$.}

%-------------------------------------------------------------------
\subsubsection{Post-Cloud Client-Side Result Recovery}
\label{subsec:result_recovery}

After cloud execution, the client selectively decrypts only the
outputs authorized for the target application. The client decrypts the returned ciphertexts results using $sk$: $$\mathsf{Dec}_{sk}(\mathsf{Enc}_{pk}(\text{\textit{result}}))$$ for \textit{result $\in$ $\{$HR, heart waveform, RR, respiration waveform, target bin N/D, ML logits$\}$}. No post-processing beyond scalar division and plaintext argmax are required.

% \medskip\noindent\textbf{Selective disclosure.} The client holds the
% sole decryption key and chooses which outputs to decrypt. A
% vital-sign monitor decrypts only $\hat{r}$ and the HR/RR BPM
% scalars; a classification client decrypts only the logit vector. All
% other information---phase signatures, Doppler spectra, intermediate
% features---remains cryptographically sealed in ciphertexts the client
% never decrypts.

%-------------------------------------------------------------------

%-------------------------------------------------------------------
\subsubsection{Example Pipelines}
\label{subsec:pipeline_vitals}
We demonstrate \name{} on two standard radar pipelines. The underlying plaintext algorithms follow Alizadeh et al.~\cite{alizadeh2019remote} for vital-signs and Li et al.~\cite{li2022di} for gesture recognition.

\noindent\textbf{Encrypted Vital Signs.}
The vital-signs pipeline composes K1$\to$K2$\to$K4$\to$K5$\to$K7$\to$K3$\to$K1$\to$K2. K1 computes per-bin energy, and K2 localizes the target bin, yielding the first decrypted output, $\hat{r}=num/den$. K4 extracts soft I/Q samples from the selected range region, K5 filters them into the respiration (0.1--0.6\,Hz) and heart-rate (0.8--2.5\,Hz) bands, and K7 recovers differential phase. Over the resulting phase trace, K3 computes a narrowband DFT, K1 computes the power spectrum $|X[k]|^2$, and K2 performs spectral sharpening and weighted frequency averaging to produce the second decrypted output, HR/RR in BPM. The total circuit depth is 11 with the third-order Taylor approximation, or 9 with the first-order approximation \rev{(Tab.~\ref{tab:depth})}. Optionally, the encrypted phase waveform after K5 can also be returned. \revdel{for client-side visualization. See Table~\ref{tab:depth} for the full depth accounting.}

\noindent\textbf{Encrypted Dynamic Radar Classification}
\label{subsec:pipeline_gesture}
The classification pipeline processes $F$ frames and composes K3$\to$K1$\to$K2 + K6 before a final encrypted FC pass. For each frame, K3 computes the encrypted range-Doppler spectrum, K1 computes its power map, and K2 together with K6 applies soft power attention with notch masking ($\gamma{=}4$) to produce weighted features. These features are accumulated across frames through homomorphic addition. The final accumulated vector is then passed through the encrypted ML network:
$\mathrm{FC}_1 \to (\cdot)^{\circ 2} \to \mathrm{FC}_2 \to (\cdot)^{\circ 2} \to \mathrm{FC}_3$. The total circuit depth is 11 \rev{(Tab.~\ref{tab:depth})}. The client decrypts the 5-class logit vector and applies plaintext argmax. \revdel{See Table~\ref{tab:depth}.}

\label{subsec:depth_budget}
\label{sec:pipeline:params}

\begin{table}[t]
  \centering
  \caption{Multiplicative depth accounting for both \name{} pipelines.
  Each pipeline encrypts the input into fresh depth-0 ciphertexts.}
  \label{tab:depth}
  \footnotesize
  \setlength{\tabcolsep}{3pt}
  \begin{tabular}{@{}lccp{0.30\columnwidth}@{}}
    \toprule
    \textbf{Stage} & \textbf{Depth} & \textbf{$\Sigma$}
      & \textbf{Operation} \\
    \midrule
    \multicolumn{4}{@{}l}{\textit{Pipeline 1: Vital Signs (depth 11, $N{=}32,768$)}} \\
    \;\;Energy integ. & 1 & 1 & $\mathrm{Re}^2{+}\mathrm{Im}^2$ \\
    \;\;Soft attention & 1 & 2 & $E_r^2$ ($\gamma{=}2$) {\scriptsize$[\to \hat{r}]$} \\
    \;\;Phase extr. & 1 & 3 & $(|z|^2)^2{\cdot}z$, sum \\
    \;\;FIR filter & 1 & 4 & pt-ct tap accumulation \\
    \;\;Taylor arctan & 3 & 7 & $y x^2{-}y^3/3$ (3rd-order) \\
    \;\;Window + DFT & 1 & 8 & pt-ct inner product \\
    \;\;$|X|^2$ & 1 & 9 & $\mathrm{Re}^2{+}\mathrm{Im}^2$ \\
    \;\;Merge + sharpen & 1 & 10 & Merge + squaring \\
    \;\;Wt.\ freq.\ avg. & 1 & 11 & pt-ct + sum {\scriptsize$[\to \mathrm{BPM}]$} \\
    \midrule
    \multicolumn{4}{@{}l}{\textit{Pipeline 2: Dynamic Classification (depth 11, $N{=}32,768$)}} \\
    \;\;Doppler DFT & 1 & 1 & BSGS block-diag matmul \\
    \;\;$|z|^2$ & 1 & 2 & $\mathrm{Re}^2{+}\mathrm{Im}^2$ \\
    \;\;Notch mask & 1 & 3 & Plaintext $\times$ cipher \\
    \;\;Soft power $\gamma{=}4$ & 2 & 5 & Two squarings \\
    \;\;Feature weighting & 1 & 6 & $p \times w$ \\
    \;\;Frame accum. & 0 & 6 & Addition only \\
    \;\;FC$_1$ + square & 2 & 8 & Matmul + $x^2$ \\
    \;\;FC$_2$ + square & 2 & 10 & Matmul + $x^2$ \\
    \;\;FC$_3$ & 1 & 11 & Matmul (final logits) \\
    \bottomrule
  \end{tabular}
\end{table}

\section{Security Analysis}
\label{sec:security}

\name{} \rev{provides} two security properties: \emph{Input Privacy}, \rev{which follows directly from the IND-CPA security of the underlying CKKS~\cite{cheon2017homomorphic} scheme}, and \emph{Data Obliviousness}, \rev{which we prove at the  arithmetic-circuit level as a property of \name{}s' kernel design}. In this section we formally prove that we achieve the security guarantees as properties of the \name{} system design.

\subsection{Input Privacy and Data Obliviousness of \name{}}
\label{subsec:input_privacy}
\textit{Input Privacy:} Since $\tilde{\mathbf{z}}[t]$ is the input to \name{}, we need to prove that the cloud learns nothing about the mmWave radar data input when it is encrypted under the CKKS scheme.

\begin{theorem}[Input Privacy guarantee of \name{}]
\label{thm:input_privacy}
Let $\mathcal{E}$ be an IND-CPA-secure CKKS scheme. For any two sequences of sensor inputs $\{\tilde{\mathbf{z}}_0[t]\}_{t=1}^{F}$ and $\{\tilde{\mathbf{z}}_1[t]\}_{t=1}^{F}$, no PPT adversary can distinguish $\{\mathsf{Enc}(\tilde{\mathbf{z}}_0[t])\}$ from $\{\mathsf{Enc}(\tilde{\mathbf{z}}_1[t])\}$ with non-negligible advantage.
\end{theorem}
\noindent
The proof (App. ~\ref{sec:app_B}) follows a hybrid argument over $F$ frames: if an adversary can distinguish the encrypted sensor streams of two different inputs, it can also break the IND-CPA security of the underlying CKKS scheme, which is assumed secure. \revdel{Full proof in Appendix~\ref{sec:app_B}.}\\

\noindent\textit{Data Obliviousness:}
\name{} protects the computation trace by design: each kernel is implemented as a fixed arithmetic circuit, and their sequential composition inherits obliviousness by a standard argument~\cite{goldreich1996software}.

\begin{theorem}[Data Obliviousness Guarantee of \name{}] 
\label{thm:data_oblivious}
Any cloud-side pipeline composed exclusively of \name{} kernels (K1--K7) and FC inference layers is data-oblivious: for any two ciphertext inputs of the same dimension, the execution trace (sequence of operations, memory addresses read and written) is
identical.
\end{theorem}
\noindent The proof (App. ~\ref{sec:app_C}) first shows that each kernel is individually data-oblivious by construction (its operation sequence depends only on public parameters, not on the encrypted input), then uses induction to show that sequential composition of oblivious kernels preserves obliviousness. \revdel{Full proof in Appendix~\ref{sec:app_C}.}

\textit{Remark on public parameters \rev{and libraries.}} The DFT kernel $\mathbf{W}$ (Eq.~\ref{eq:dft_kernel}), notch mask $m[d]$ (Eq.~\ref{eq:notch_mask}), FIR filter taps $\mathbf{T}^{(b)}$ (Eq.~\ref{eq:fir_iq}), and FC weights $\mathbf{W}^{(\ell)}, \mathbf{b}^{(\ell)}$ (Eq.~\ref{eq:mlp_forward}) are public model parameters shared with the cloud. They contain no user-specific information and do not require privacy protection. \revdel{Input Privacy protects the sensor data, not the processing circuit.} \rev{Open-source FHE backends executing these circuits during implementation: OpenFHE~\cite{al2022openfhe} and FIDESlib~\cite{agullo2025fideslib} are assumed to preserve the CKKS security properties at runtime, we do not verify the libraries for correctness.
% i.e., that rescaling, key switching, memory allocation, and GPU kernel scheduling do not depend on ciphertext contents. We do not verify the libraries themselves for correctness.
}

%-------------------------------------------------------------------
\subsection{Leakage Profile}
\label{subsec:leakage}

We explicitly quantify the residual information exposed to the cloud.

\noindent\textbf{Leakage~1: Protocol interaction pattern.} The protocol is unidirectional per window: client sends encrypted data, cloud evaluates, client decrypts. Ciphertext count and size are fixed by public configuration, not input data. All ciphertexts are IND-CPA secure (Theorem~\ref{thm:input_privacy}).

\noindent\textbf{Leakage~2: Public configuration metadata.} The cloud observes $R, D, A, F, \gamma$, filter order, Taylor order, FC architecture, and session timing; all fixed prior to deployment.

\noindent\textbf{Leakage~3: Public model parameters.} DFT kernel $\mathbf{W}$, notch mask $m[d]$, FIR taps $\mathbf{T}^{(b)}$, and FC weights $\mathbf{W}^{(\ell)}, \mathbf{b}^{(\ell)}$---analogous to circuit descriptions in standard FHE protocols.

In summary, total leakage consists exclusively of public metadata and model parameters. No range-bin index, phase value, energy distribution, or vital-sign waveform is exposed.

%-------------------------------------------------------------------
\subsection{Defense Against CKKS-Specific Attacks}
\label{subsec:ckks_defenses}

The CKKS scheme introduces unique security considerations due to
its approximate arithmetic semantics.

\noindent\textbf{Noise-growth analysis.}
In branching circuits, CKKS noise variance can differ depending on
the computation path, potentially leaking the branch condition.
Because \name{} executes a fixed circuit with no branching
(Theorem~\ref{thm:data_oblivious}), noise growth is identical
for all inputs at every ciphertext level. An adversary analyzing
output noise variance learns nothing beyond the public parameters.

% \noindent\textbf{Approximate arithmetic.}
% Each CKKS rescaling introduces rounding error $O(2^{-\Delta})$.
% For the vital signs pipeline ($\Delta{=}40$\,bits, $L{=}11$),
% cumulative error is
% ${\sim}\,14 \times 2^{-40} \approx 1.3 \times 10^{-11}$. For the
% dynamic classification pipeline ($\Delta{=}50$\,bits, $L{=}11$),
% cumulative error is
% ${\sim}\,11 \times 2^{-50} \approx 9.8 \times 10^{-15}$. Both are
% orders of magnitude below the radar sensor noise floor (typically
% ${>}\,10^{-6}$ in normalized units). The rounding artifacts do not
% create a distinguishable side channel.
\noindent\textbf{Approximate arithmetic.}
Each CKKS rescaling introduces rounding error $O(2^{-\Delta})$. For the vital-signs pipeline ($\Delta{=}40$\,bits), cumulative error is ${\sim}\,1.3 \times 10^{-11}$. For the dynamic classification pipeline ($\Delta{=}50$\,bits), cumulative error is ${\sim}\,9.8 \times 10^{-15}$. Both are orders of magnitude below the radar sensor noise floor (typically ${>}\,10^{-6}$ in normalized units). The rounding artifacts do not create a distinguishable side channel.

\noindent\textbf{IND-CPA$^D$ and passive decryption attacks.}
Li and Micciancio~\cite{li2021security} showed that CKKS is
vulnerable to chosen-ciphertext attacks where an adversary recovers
the secret key from approximate decryption results. In \name{}, the
cloud never receives decrypted outputs: all decryption occurs on the
trusted client; so the IND-CPA$^D$ attack vector does not apply.

\section{Evaluation}
\label{sec:evaluation}

% We evaluate \name{} along four axes: per-kernel fidelity
% (\S\ref{subsec:fidelity}), end-to-end accuracy
% (\S\ref{subsec:accuracy}), computational overhead
% (\S\ref{subsec:overhead}), and privacy gain
% (\S\ref{subsec:privacy_eval}).

\subsection{Setup}
\label{subsec:setup}
\label{sec:implementation}

\noindent\textbf{Datasets.}
We evaluate \name{} on three public FMCW radar datasets spanning two sensing applications. Table~\ref{tab:dataset_config} summarizes the radar configuration for each dataset. For vital-signs we use Children~\cite{children_dataset} (pediatric subjects) and Parralejo~\cite{parralejo2026mmwave} (lying/sitting conditions). For gesture recognition we use \rev{the} 60\,GHz \rev{gesture dataset}~\cite{zenodo_gesture} (21{,}000 samples, 5 classes, 6 environments). \rev{All three datasets record a single dominant subject in the line of sight of the radar: the vital-sign datasets capture one stationary participant seated or lying ${\sim}0.5$\,m from the sensor in a single furnished room, and the gesture dataset captures one participant gesturing within 1\,m of a short-range radar  across six indoor environments (meeting room, open-plan office, library, kitchen, exercise
room, bedroom). Scenes with multiple people or non-line-of-sight propagation are not represented in this evaluation.}

\begin{table}[t]
  \centering
  \caption{Radar and dataset configurations. All datasets are publicly
  available.}
  \label{tab:dataset_config}
  \footnotesize
  \setlength{\tabcolsep}{1pt}
  \begin{tabular}{@{}l l cccccc@{}}
    \toprule
    \textbf{Dataset} & \textbf{Task}
      & \textbf{Radar} & $f_c$ & \textbf{BW}
      & \textbf{FPS} & \textbf{Subj.} & \textbf{GT} \\
    & & & (GHz) & (GHz) & (Hz) & & \\
    \midrule
    Children~\cite{children_dataset}
      & Vital & IWR6843 & 60 & 3.75 & 20 & 50 & BSM6501K \\
    Parralejo~\cite{parralejo2026mmwave}
      & Vital & IWR6843ISK & 60.25 & 0.48 & 10 & 110 & ECG \\
    \rev{Gesture}~\cite{zenodo_gesture}
      & Gesture & BGT60TR13C & 60.5 & 4 & 33 & \rev{12} & Labels \\
    \bottomrule
  \end{tabular}
\end{table}

\noindent\textbf{Processing configuration.}
For the vital-signs pipeline we use $R{=}64$ range bins and $F{=}200$ frames, corresponding to a 10\,s sample window for Children (20\,Hz) and a 20\,s sample window for Parralejo (10\,Hz), yielding 6\,bpm and 3\,bpm frequency resolution, respectively. For the gesture pipeline, we retain $R{=}16$ range bins and $D{=}32$ chirps across $A{=}3$ antennas, yielding $ARD{=}1{,}536$ active slots per ciphertext out of $N/2{=}\rev{16{,}384}$ available.
\rev{In both pipelines, the kernels aggregate per-frame statistics within each fixed sample window. The current implementation does not track any \textit{state} across frames. We also note that our example pipelines do not include constructing or processing 3D heatmaps or point clouds as intermediate steps (\S\ref{sec:discussion}).}

\noindent\textbf{Baselines.} We compare \name{} against the following three baselines.
(1)~\textit{Standard DSP on Plaintext}: standard DSP algorithms executed on unencrypted data;
% , establishing the conventional accuracy ceiling.
(2)~\textit{\name{} on Plaintext}: identical DSP kernels executed on unencrypted data; and
% , establishing the error introduced by the \name{} kernel approximations of standard DSP algorithms.
(3)~\textit{RPRS~\cite{wu2025rprs} on plaintext feature maps}: encrypted post-processed radar features with FHE CNN inference.

\noindent\textbf{Metrics.}
For vital-signs application pipeline we compare \name{} against baseline (2) described above; HR/RR MAE is the mean absolute difference in the estimated rate (bpm), waveform MSE is the mean squared error of the extracted phase signal per window, and latency is GPU wall-clock time per second of sensor data recording. \rev{These metrics isolate the error introduced by encryption and should not be read as clinical accuracy.}
For gesture classification we report top-1 accuracy under 6-fold leave-one-environment-out cross-validation\rev{, over the 12{,}000
nominal recordings of one subject (2{,}000 per environment).}

\noindent\textbf{Training.}
The FC networks are trained offline in plaintext
using the exact feature computation required for inference. After convergence the weights are frozen and shipped to
the cloud as public plaintext parameters.

\noindent\textbf{Implementation.}
We implement both application pipelines using
OpenFHE~\cite{al2022openfhe} (CPU, Python/C++) and FIDESlib~\cite{agullo2025fideslib} (GPU, C++) CKKS libraries with parameters listed in Table~\ref{tab:depth}.
All experiments run on a single workstation with an AMD Ryzen~9 5950X
(16C/32T), 64\,GB RAM, and an NVIDIA RTX~3090\,Ti (24\,GB, CUDA~12.9).
Client-side encryption and decryption are run on the same machine, as well as Raspberry Pi 4 (4~GB) for benchmarking.

%% ----------------------------------------------------------------
\subsection{Kernel Fidelity and Noise Budget}
\label{subsec:fidelity}

Table~\ref{tab:kernel_fidelity} reports the per-kernel
\emph{encryption noise}: the MSE between encrypted and plaintext execution of the identical pipeline, isolating the error introduced
by CKKS ciphertext arithmetic alone.
Each row is measured cumulatively (i.e., after all preceding kernels),
over 3 test recordings per pipeline.
All kernels achieve encryption noise below $10^{-5}$, confirming that
CKKS arithmetic does not meaningfully perturb kernel outputs at any
stage.

\begin{table}[t]
  \centering
  \caption{Per-kernel encryption noise (encrypted vs.\ plaintext
  execution of the identical pipeline).}
  \label{tab:kernel_fidelity}
  \footnotesize
  \setlength{\tabcolsep}{3pt}
  \begin{tabular}{@{}l ccc@{}}
    \toprule
    \textbf{Kernel} & \textbf{MSE} & \textbf{Max $|$err$|$}
      & \textbf{Depth} \\
    \midrule
    K1: Energy integ.     & $4.1{\times}10^{-14}$ & $1.5{\times}10^{-6}$ & 1 \\
    K2: Soft attention     & $5.5{\times}10^{-6}$ & $6.9{\times}10^{-3}$ & 2 \\
    K3: Doppler DFT        & $4.8{\times}10^{-22}$ & $7.8{\times}10^{-11}$ & 1 \\
    K4: Phase extr.        & $4.6{\times}10^{-14}$ & $1.0{\times}10^{-6}$ & 3 \\
    K5: FIR filter         & $1.5{\times}10^{-14}$ & $5.2{\times}10^{-7}$ & 4 \\
    K6: Notch mask         & $8.2{\times}10^{-24}$ & $7.0{\times}10^{-12}$ & 3 \\
    K7: Taylor arctan       & $5.3{\times}10^{-15}$ & $2.3{\times}10^{-7}$ & 5 \\
    \bottomrule
  \end{tabular}
\end{table}

\noindent\textbf{FHE-friendly approximation fidelity (Baseline 1).}
Of the seven kernels, four (K1, K3, K5, K6) are mathematically
equivalent to their standard DSP counterparts: energy summation,
DFT via matrix multiply, FIR convolution via Toeplitz multiply, and
elementwise masking introduce no approximation error beyond CKKS noise.
The remaining kernels replace non-polynomial operations with
FHE-compatible alternatives: K4 uses soft weighted I/Q extraction
instead of hard peak selection, and K7 uses a first-order Taylor
cross-product instead of \texttt{atan2} followed by phase unwrapping
and differentiation.
To quantify this trade-off, we compare the phase signals produced by
the FHE-friendly pipeline (K4+K7) against a standard DSP reference
(hard argmax $\to$ \texttt{atan2} $\to$ unwrap $\to$ differentiation)
across 40 recordings from three datasets.
The mean phase-signal MSE is 0.153 for the respiratory band and
0.182 for the heart-rate band, yet the end-to-end HR/RR estimates
remain within ${<}10^{-3}$\,bpm of the plaintext baseline
(Tab.~\ref{tab:vital_signs_accuracy}), confirming that the
polynomial approximations are a practical trade-off for FHE
compatibility.

Figure~\ref{fig:noise_budget} tracks the encryption noise at each
kernel output as multiplicative depth grows through the pipeline.
For the vital-signs pipeline ($L{=}11$), K2 exhibits the highest
encryption noise (${\sim}4{\times}10^{-6}$) because squaring large
energy values amplifies the output magnitude; subsequent kernels drop
back to ${\sim}10^{-14}$ as K4's plaintext rescaling
($1/\text{frame\_max}$ via plaintext$\times$ciphertext multiplication)
contracts the output range.
The final PSD output at depth~10 reaches ${\sim}5{\times}10^{-12}$.
For the gesture pipeline ($L{=}11$), encryption noise remains below
$10^{-21}$ through the DSP stages and reaches ${\sim}5{\times}10^{-19}$
at the FC output, where matrix--vector accumulation over thousands of
terms raises the noise floor.
The bars are not monotonically increasing because absolute CKKS error
scales with output magnitude, not just multiplicative depth: kernels
that contract the value range (e.g., plaintext rescaling in K4, small
FIR taps in K5) reduce the absolute noise envelope, while operations
that amplify magnitude (K2 squaring, FC matmul) increase it.
In both pipelines the encryption noise remains orders of magnitude
below the signal, confirming that the CKKS noise budget is comfortably
within operating margins.

\begin{figure}[]
  \centering
  \includegraphics[width=\columnwidth]{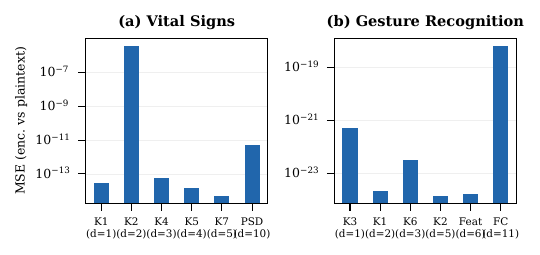}
  \caption{Per-kernel encryption noise (MSE between encrypted and
  plaintext execution of the identical pipeline) at each stage.
  Bars are not monotonically increasing: kernels that contract the value range
  (e.g., plaintext rescaling in K4) reduce absolute error, while those
  that amplify it (K2 squaring, FC matmul) increase it.}
  \label{fig:noise_budget}
\end{figure}

%% ----------------------------------------------------------------
\subsection{Accuracy \rev{of \name{}}}
\label{subsec:accuracy}

\noindent\textbf{Vital Signs.}
Table~\ref{tab:vital_signs_accuracy} compares the plaintext and encrypted pipelines across two vital-signs datasets comprising 270 total recordings. All metrics measure the \emph{discrepancy} between the plaintext baseline (2) and the \name{} encrypted pipeline \rev{(not to be read as clinical accuracy)}. The encrypted pipeline produces outputs indistinguishable from the baseline at the rate level ($< 10^{-3}$\,bpm), with waveform MSE dominated by the first-order Taylor  approximation rather than CKKS noise.

\begin{table}[]
  \centering
  \caption{Vital-signs fidelity: waveform MSE and rate MAE,
  plaintext vs.\ encrypted pipeline. $^\dagger$Lying, $^\ddagger$Sitting.}
  \label{tab:vital_signs_accuracy}
  \footnotesize
  \setlength{\tabcolsep}{3pt}
  \begin{tabular}{@{}l c cc c@{}}
    \toprule
    \textbf{Dataset} & \textbf{$n$}
      & \textbf{HR MSE}\,$\downarrow$
      & \textbf{RR MSE}\,$\downarrow$
      & \textbf{HR/RR Rate MAE}\,$\downarrow$ \\
    \midrule
    Children  & 50 & 0.264 & 0.181 & ${<}10^{-3}$\,bpm \\
    Parralejo & 110$^\dagger$/110$^\ddagger$
              & 0.064$^\dagger$/0.136$^\ddagger$
              & 0.074$^\dagger$/0.179$^\ddagger$
              & ${<}10^{-3}$\,bpm \\
    \bottomrule
  \end{tabular}
\end{table}

\noindent\textbf{Gesture Classification.}
Table~\ref{tab:gesture_accuracy} reports per-environment classification
accuracy on the 5-class gesture dataset~\cite{zenodo_gesture}, measured against baseline (2) over 100 test recording per fold. The plaintext pipeline achieves 84.7\% mean accuracy under 6-fold cross-validation; the encrypted pipeline matches at 84.5\% with 99.8\% prediction agreement, as the CKKS noise floor ($\sim\!10^{-10}$) does not perturb the argmax of the output logits in all but one recording.

% \begin{table}[t]
%   \centering
%   \caption{Gesture classification accuracy (5-class, 60\,GHz \rev{gesture dataset}).
%   6-fold leave-one-environment-out cross-validation.}
%   \label{tab:gesture_accuracy}
%   \footnotesize
%   \setlength{\tabcolsep}{4pt}
%   \begin{tabular}{@{}l ccc@{}}
%     \toprule
%     \textbf{Fold} & \textbf{Plaintext} & \textbf{mmFHE} & \textbf{Agreement} \\
%     \midrule
%     e1 & 72.0\% & 72.0\% & 100\% \\
%     e2 & 88.0\% & 88.0\% & 100\% \\
%     e3 & 82.0\% & 82.0\% & 100\% \\
%     e4 & 92.0\% & 92.0\% & 100\% \\
%     e5 & 86.0\% & 85.0\% & 99\% \\
%     e6 & 88.0\% & 88.0\% & 100\% \\
%     \midrule
%     \textbf{Mean} & \textbf{84.7\%} & \textbf{84.5\%} & \textbf{99.8\%} \\
%     \bottomrule
%   \end{tabular}
% \end{table}

\begin{table}[]
  \centering
  \caption{Gesture classification accuracy (5-class, 60\,GHz
  \rev{gesture dataset}). 6-fold leave-one-environment-out
  cross-validation.}
  \label{tab:gesture_accuracy}
  \footnotesize
  \setlength{\tabcolsep}{3.5pt}
  \begin{tabular}{@{}l ccccccc@{}}
    \toprule
    \textbf{Fold} & \textbf{e1} & \textbf{e2} & \textbf{e3} & \textbf{e4}
      & \textbf{e5} & \textbf{e6} & \textbf{Mean} \\
    \midrule
    \textbf{Plaintext} & 72\% & 88\% & 82\% & 92\% & 86\% & 88\% & \textbf{84.7\%} \\
    \textbf{\name{}}   & 72\% & 88\% & 82\% & 92\% & 85\% & 88\% & \textbf{84.5\%} \\
    \midrule
    Agreement & 100\% & 100\% & 100\% & 100\% & 99\% & 100\% & \textbf{99.8\%} \\
    \bottomrule
  \end{tabular}
\end{table}

%% ----------------------------------------------------------------
\subsection{Computational Overhead}
\label{subsec:overhead}

Table~\ref{tab:latency} breaks down per-stage GPU latency for both pipelines \rev{under the per-frame and the batched packing strategies. Context setup and key generation are one-time per-deployment costs ($^\dagger$), reported in the table but excluded from the per-window totals and slowdown factors. We report steady-state total costs.}

% \begin{table}[t]
%   \centering
%   \caption{Per-stage latency breakdown. Plaintext times in ms;
%   encrypted times in seconds.}
%   \label{tab:latency}
%   \footnotesize
%   \setlength{\tabcolsep}{3pt}
%   \begin{tabular}{@{}l r r@{}}
%     \toprule
%     \textbf{Stage}
%       & \textbf{Plain (ms)}
%       & \textbf{GPU (s)} \\
%     \midrule
%     \multicolumn{3}{@{}l}{\textit{Pipeline 1: Vital Signs (Children, 200\,fr @ 20\,Hz)}} \\
%     \;\;Context + keygen     & ---   & 3.84 \\
%     \;\;Encrypt + K1: Energy & 0.119  & 11.66 \\
%     \;\;K2: Soft attention    & 0.010  & 0.05 \\
%     \;\;Encrypt + K4: Phase  & 2.551   & 15.91 \\
%     \;\;K5+K7: FIR + Taylor  & 0.016   & 43.22 \\
%     \;\;DFT + PSD + decrypt  & 0.094   & 28.64 \\
%     \cmidrule(l){2-3}
%     \;\;\textbf{Total}       & \textbf{2.790}  & \textbf{103.32} \\
%     \;\;Slowdown             & $1\times$    & ${37{,}033\times}$ \\
%     \midrule
%     \multicolumn{3}{@{}l}{\textit{Pipeline 2: Dynamic Classification (100\,fr @ 33\,Hz)}} \\
%     \;\;Context + keygen     & ---   & 3.69 \\
%     \;\;Encrypt              & ---   & 6.17 \\
%     \;\;DSP (K3+K1+K6+K2+feat.)& 13.298 & 20.96 \\
%     \;\;FC1 + sq             & 0.138  & 4.43 \\
%     \;\;FC2 + sq             & 0.109  & 3.91 \\
%     \;\;FC3                  & 0.001  & 1.09 \\
%     \;\;Decrypt              & ---   & 0.01 \\
%     \cmidrule(l){2-3}
%     \;\;\textbf{Total}       & \textbf{13.55} & \textbf{36.56} \\
%     \;\;Slowdown             & $1\times$ & ${2{,}699\times}$ \\
%     \bottomrule
%   \end{tabular}
% \end{table}
\begin{table}[t]
  \centering
  \caption{Per-stage latency breakdown. \rev{$^\dagger$One-time setup per deployment, excluded from totals; $^\ddagger$includes one-time encoding of the public plaintext operands.}}
  \label{tab:latency}
  \footnotesize
  \setlength{\tabcolsep}{3pt}
  \begin{tabular}{@{}l r r r@{}}
    \toprule
    \textbf{Stage}
      & \textbf{Plain}
      & \rev{\textbf{mmFHE}}
      & \rev{\textbf{mmFHE}} \\
      & \textbf{(ms)}
      & \rev{\textbf{per-frame (s)}}
      & \rev{\textbf{batched (s)}} \\
    \midrule
    \multicolumn{4}{@{}l}{\textit{Pipeline 1: Vital Signs (Children, 200\,fr @ 20\,Hz)}} \\
    \;\;Context + keygen\rev{$^\dagger$}     & ---   & 3.84 & \rev{1.98$^\ddagger$} \\
    \;\;Encrypt + K1: Energy & 0.119  & 11.66 & \rev{0.07} \\
    \;\;K2: Soft attention    & 0.010  & 0.05 & \rev{0.04} \\
    \;\;Encrypt + K4: Phase  & 2.551   & 15.91 & \rev{0.13} \\
    \;\;K5+K7: FIR + Taylor  & 0.016   & 43.22 & \rev{0.03} \\
    \;\;DFT + PSD + decrypt  & 0.094   & 28.64 & \rev{0.94} \\
    \cmidrule(l){2-4}
    \;\;\textbf{Total}       & \textbf{2.790}  & \rev{\textbf{99.48}} & \rev{\textbf{1.21}} \\
    \;\;Slowdown             & $1\times$    & \rev{$35{,}656\times$} & \rev{$434\times$} \\
    \midrule
    \multicolumn{4}{@{}l}{\textit{Pipeline 2: Dynamic Classification (100\,fr @ 33\,Hz)}} \\
    \;\;Context + keygen\rev{$^\dagger$}     & ---   & 3.69 & \rev{3.75$^\ddagger$} \\
    \;\;Encrypt              & ---   & 6.17 & \rev{0.82} \\
    \;\;DSP (K3+K1+K6+K2+feat.)& 13.298 & 20.96 & \rev{0.47} \\
    \;\;FC1 + sq             & 0.138  & 4.43 & \rev{2.27} \\
    \;\;FC2 + sq             & 0.109  & 3.91 & \rev{2.06} \\
    \;\;FC3                  & 0.001  & 1.09 & \rev{0.13} \\
    \;\;Decrypt              & ---   & 0.01 & \rev{0.01} \\
    \cmidrule(l){2-4}
    \;\;\textbf{Total}       & \textbf{13.55} & \textbf{36.56} & \rev{\textbf{5.76}} \\
    \;\;Slowdown             & $1\times$ & ${2{,}699\times}$ & \rev{$425\times$} \\
    \bottomrule
  \end{tabular}
\end{table}

\noindent\rev{In the per-frame packing,} FIR+Taylor arctan dominates the
vital-signs pipeline (\rev{43\%} of
total), followed by PSD peak detection (\rev{29}\%) and phase extraction
(\rev{16}\%). Both stages perform per-frame ciphertext--ciphertext
multiplications over all $F{=}200$ frames.
The \rev{per-frame} PSD stage cost scales linearly with the number of frequency
bins searched: at 10\,Hz frame rate (Parralejo), the HR band
contains \rev{34} DFT bins vs.\ \rev{17} at 20\,Hz (Children), doubling
PSD time to 57\,s and raising total latency to 132\,s.
For the gesture pipeline (3-second windows), the GPU achieves
36.6\,s per inference.

\begin{revblock}
\noindent\textbf{Batched execution.}
The batched packing (\S\ref{subsec:client_preprocessing}) encrypts the full vital-signs window into 4 ciphertexts and the gesture window into 26 (8 frames per ciphertext), so every kernel processes all frames simultaneously in SIMD fashion. The public plaintext operands are encoded once per deployment and cached ($^\ddagger$ in Tab.~\ref{tab:latency}). In steady state this reduces the end-to-end latency to 1.21\,s per vital-signs window ($434\times$ over plaintext) and 5.76\,s per gesture window ($425\times$). The previous per-frame bottlenecks become negligible, leaving the vital-signs pipeline dominated by the PSD stage (0.94\,s) and the gesture pipeline by the FC layers (4.5\,s). Because the batched PSD evaluates the partial DFT as a single diagonal matrix product over the packed phase series, its cost is nearly independent of the number of bins searched: Parralejo dataset (34 HR bins) completes in 0.9--1.1\,s per 20\,s window, on par with Children (17 bins). 
% Batching changes only the ciphertext layout, not the computation: the outputs match the per-frame pipeline within output precision on all 270 vital-sign recordings, with identical per-fold gesture accuracy (Table~\ref{tab:gesture_accuracy}).
\end{revblock}

\noindent\textbf{Amortized cost.}
\rev{In the per-frame packing,} the Children pipeline processes 200 frames (10\,s at 20\,Hz) in \rev{99.5}\,s on GPU, yielding an amortized cost of \rev{497}\,ms per frame (\rev{9.9}\,s per second of sensor data). The Parralejo pipeline processes 200 frames (20\,s at 10\,Hz) in 132\,s, yielding 660\,ms per frame (6.6\,s per second). The gesture pipeline \rev{processes} $F{=}100$ frames (3\,s at 33\,Hz) per \rev{window}, giving 366\,ms per frame. \rev{With batching, the amortized cost falls to 6\,ms per frame (0.12\,s per second of sensor data) for Children, $5$\,ms for Parralejo, and 58\,ms for gesture.}

\noindent\textbf{Communication and resource overhead.} Table~\ref{tab:comm_overhead} reports per-window communication costs, client-side encryption latency, and cloud GPU memory. Both pipelines use ring dimension $N{=}32{,}768$ ($128$-bit security), yielding ciphertexts of 5.5--6.0\,MB each. \rev{In the per-frame packing,} the vital-signs pipeline encrypts 800 ciphertexts per 10\,s window (\rev{four per frame: a Re/Im pair for each of the two input streams, K1 and K4}), totalling 4.3\,GB uplink. The gesture pipeline encrypts 200 ciphertexts per 3\,s window (1.2\,GB uplink). Key material (public, relinearization, and Galois rotation keys) is a one-time setup cost of 3.1\,GB (vital, 140 rotation indices) and \rev{4.4}\,GB (gesture, 188 indices); in practice, this is transmitted once at enrollment and cached on the cloud. \rev{We measure client-side encryption on a Raspberry Pi~4 edge device (4\,GB, single core): in the per-frame packing it takes 164\,s per vital-signs window and 44.3\,s per gesture window ($\sim$205--228\,ms per ciphertext), infeasible for live operation.} The cloud GPU (RTX 3090 Ti, 24\,GB) peaks at 20.7\,GB during \rev{per-frame} vital-signs inference and 4.3\,GB during gesture inference. \rev{Batching reduces all client-side costs: the input collapses to 4 and 26 ciphertexts per window, cutting the uplink to 22\,MB and 156\,MB; the one-time key material drops to 1.4\,GB (60 rotation indices) and 2.7\,GB (115 indices); and encryption on the Raspberry Pi~4 falls to 0.84\,s and 5.9\,s per window, making the edge client viable. On the cloud side, the batched layouts peak at 5.3\,GB and 7.5\,GB of GPU memory, most of which holds the cached encodings of the public constants (DFT diagonals, filter taps, FC weights) reused across windows.}

\begin{table}[t]
  \centering
  \caption{Communication and resource overhead per sensing window \revdel{Key material is a one-time setup cost; uplink/downlink are per window.} \rev{for a Raspberry Pi~4 client.}}
  \label{tab:comm_overhead}
  \footnotesize
  \setlength{\tabcolsep}{2pt}
  \begin{tabular}{@{}l rrrr@{}}
    \toprule
    \textbf{Metric}
      & \multicolumn{2}{c}{\textbf{Vital Signs}}
      & \multicolumn{2}{c@{}}{\textbf{Gesture}} \\
    \cmidrule(lr){2-3}\cmidrule(l){4-5}
      & \rev{\textbf{per-fr.}} & \rev{\textbf{batch.}}
      & \rev{\textbf{per-fr.}} & \rev{\textbf{batch.}} \\
    \midrule
    Ring dimension $N$       & 32{,}768 & \rev{32{,}768} & 32{,}768 & \rev{32{,}768} \\
    Ciphertext size          & 5.5\,MB  & \rev{5.5\,MB} & 6.0\,MB & \rev{6.0\,MB} \\
    Input cts / window       & 800      & \rev{4} & 200 & \rev{26} \\
    Uplink / window          & 4.3\,GB  & \rev{22\,MB} & 1.2\,GB & \rev{156\,MB} \\
    Downlink / window        & \rev{33\,MB} & \rev{33\,MB} & 6\,MB & \rev{6\,MB} \\
    \midrule
    Key material (one-time)  & 3.1\,GB  & \rev{1.4\,GB} & \rev{4.4\,GB} & \rev{2.7\,GB} \\
    \;\;Public key           & 5.5\,MB  & \rev{5.5\,MB} & 6.0\,MB & \rev{6.0\,MB} \\
    \;\;Relin.\ key          & 22.5\,MB & \rev{22.5\,MB} & \rev{24.0\,MB} & \rev{24.0\,MB} \\
    \;\;Galois keys          & 3.1\,GB  & \rev{1.3\,GB} & \rev{4.4\,GB} & \rev{2.7\,GB} \\
    \midrule
    \rev{Client encrypt (RPi\,4)} & \rev{164\,s} & \rev{0.84\,s} & \rev{44.3\,s} & \rev{5.9\,s} \\
    \rev{Encrypt / ct (RPi\,4)}   & \rev{205\,ms} & \rev{210\,ms} & \rev{222\,ms} & \rev{228\,ms} \\
    Cloud GPU memory (peak)  & 20.7\,GB & \rev{5.3\,GB} & 4.3\,GB & \rev{7.5\,GB} \\
    \bottomrule
  \end{tabular}
\end{table}

% \begin{table}[t]
%   \centering
%   \caption{Communication and resource overhead per sensing window.
%   Key material is a one-time setup cost; uplink/downlink are per window.}
%   \label{tab:comm_overhead}
%   \footnotesize
%   \setlength{\tabcolsep}{3pt}
%   \begin{tabular}{@{}l rr@{}}
%     \toprule
%     \textbf{Metric}
%       & \textbf{Vital Signs}
%       & \textbf{Gesture} \\
%     \midrule
%     Ring dimension $N$       & 32{,}768 & 32{,}768 \\
%     Ciphertext size          & 5.5\,MB  & 6.0\,MB \\
%     Input cts / window       & 800      & 200 \\
%     Uplink / window          & 4.3\,GB  & 1.2\,GB \\
%     Downlink / window        & 11\,MB   & 6\,MB \\
%     \midrule
%     Key material (one-time)  & 3.1\,GB  & 7.5\,GB \\
%     \;\;Public key           & 5.5\,MB  & 6.0\,MB \\
%     \;\;Relin.\ key          & 22.5\,MB & 46.5\,MB \\
%     \;\;Galois keys          & 3.1\,GB  & 7.5\,GB \\
%     \midrule
%     Client encrypt (CPU, 1 core) & 26.1\,s & 6.7\,s \\
%     Client encrypt / ct      & 33\,ms   & 33\,ms \\
%     Cloud GPU memory (peak)  & 20.7\,GB & 4.3\,GB \\
%     \bottomrule
%   \end{tabular}
% \end{table}

%% ----------------------------------------------------------------
\subsection{Privacy Analysis}
\label{subsec:privacy_eval}

\noindent\textbf{mmFHE vs.\ RPRS.}
Table~\ref{tab:rprs_comparison} compares \name{} against the RPRS~\cite{wu2025rprs}  baseline (3) and plaintext baseline (1). Since RPRS's original dataset is not publicly available, we report their CNN-only FHE approach on the public 60\,GHz gesture dataset\rev{~\cite{zenodo_gesture}}\revdel{for a fair comparison}.

RPRS achieves the same accuracy as the plaintext baseline (95.0\%),
since its CNN input is computed in cleartext.
\name{} achieves 84.5\% under full encryption; however, our plaintext
pipeline (identical DSP and CNN, no encryption) reaches 84.67\%,
confirming that CKKS introduces less than 0.2\,percentage points accuracy loss.
The remaining gap to 95.0\% stems from the different feature extraction
pipelines, not from encryption\revdel{error}.

In terms of latency, \rev{in the per-frame packing} the encrypted DSP
stage accounts for the bulk of \name{}'s overhead (20.96\,s per
inference)\rev{; batching reduces the DSP stage to 0.47\,s, leaving the
FC layers (4.5\,s) as the dominant cost}.
Notably, \name{}'s CNN-only FHE latency (9.42\,s \rev{per-frame,
4.5\,s batched}) is lower than RPRS's (\rev{14.09}\,s)
despite a deeper circuit (depth~11 vs.~5), because of GPU acceleration
\rev{and cached plaintext operands}.
\rev{With batching, the entire encrypted \name{} pipeline (5.76\,s,
Table~\ref{tab:latency}) runs $2.4\times$ faster than RPRS's CNN-only
stage alone, while encrypting the full DSP chain rather than only the
classifier.}
Table~\ref{tab:fhe_params} \rev{summarizes} the CKKS
parameters of both systems \rev{(entries marked ``---'' were not reported)}.

\begin{table}[t]
  \centering
  \caption{Three-way gesture classification comparison on the \rev{gesture} dataset (6-fold CV, $n{=}100$ per fold).
  }
  \label{tab:rprs_comparison}
  \footnotesize
  \begin{tabular}{@{}lcc@{}}
    \toprule
    \textbf{Method} & \textbf{Accuracy} & \textbf{Latency}
      \\
    \midrule
    Plaintext      & 95.0\% & 0.04\,ms \\
    RPRS (repro.)  & 95.0\% & 14.09\,s  \\
    \name{} (ours) & 84.5\% & \rev{5.76}\,s \\
    \bottomrule
  \end{tabular}
\end{table}

\begin{table}[t]
  \centering
  \caption{CKKS parameters of RPRS~\cite{wu2025rprs} and  \name{}.}
  \label{tab:fhe_params}
  \footnotesize
  \begin{tabular}{@{}lcc@{}}
    \toprule
    \textbf{Parameter} & \textbf{RPRS} & \textbf{\name{} (ours)} \\
    \midrule
    FHE scheme          & CKKS              & CKKS \\
    Library             & MS SEAL 4.1 (CPU) & OpenFHE + FIDESlib (GPU) \\
    Polynomial degree   & 8{,}192           & 32{,}768 \\
    Slots               & 4{,}096           & \rev{16{,}384} \\
    Encrypted stages    & CNN only          & DSP + CNN \\
    Multiplicative depth& ---               & 11 \\
    Security level      & ---               & 128-bit classic \\
    Scaling technique   & ---               & \textsc{FlexibleAuto} \\
    \bottomrule
  \end{tabular}
\end{table}

%% ----------------------------------------------------------------

\section{Related Work}
\label{sec:related_work}

\noindent\textbf{Privacy Risks of Wireless Sensing.}
mmWave FMCW radar is now a mature modality for contactless monitoring of vital signs~\cite{adib2015smart,vilesov2022blending}, sleep~\cite{zhao2017learning}, gait~\cite{meng2020gait}, gesture~\cite{li2022towards, li2022di}, falls~\cite{jin2020mmfall}, and elderly care~\cite{guo2021design,alhazmi2024intelligent}. However, the same signal richness that enables these applications creates serious privacy risks: mmWave radar can eavesdrop on phone calls via vocal-cord vibrations~\cite{basak2022mmspy}, re-identify individuals from gait signatures~\cite{vandersmissen2018indoor}, and infer sensitive attributes from beam patterns~\cite{xiao2025lend}. Wi-Fi signals pose similar threats, enabling 3D person re-identification~\cite{ren2023person}. A recent systematization of 169 wireless sensing privacy papers~\cite{sun2022sok} concludes that effective defenses for protecting raw sensor data during cloud-side computation remain largely absent.

\noindent\textbf{Privacy Defenses for Wireless Sensing.}
Existing defenses fall into three families, none of which provides cryptographic data-in-use confidentiality for outsourced computation.

\noindent\textit{Data perturbation:} Differential privacy destroys the submillimeter phase precision that vital sign \rev{detection methods} demand~\cite{wang2025privacy}, while training-time poisoning \revdel{e.g., Poison to Cure} protects learned models but not the raw sensor streams ~\cite{hu2025poison}.

\noindent\textit{Physical-layer:} MIMOCrypt~\cite{luo2024mimocrypt} uses MIMO channel encryption, while VitalHide~\cite{gao2025vitalhide} and PrivyWave~\cite{gao2025privywave} use decoy signal injection, but both require physical-layer control over the sensing environment and do not address cloud outsourcing.

\noindent\textit{Federated learning:} It protects training data, not inference-time confidentiality of streaming sensors~\cite{sun2022sok}, which is the privacy threat primarily considered in \name{}.
% \name{} fills the orthogonal gap: cryptographic confidentiality of raw data \textit{during} computation, without hardware trust assumptions or accuracy degradation.

\noindent\textbf{Homomorphic Encryption for Sensing and Inference.}
FHE has been applied to protect ML inference, but prior work encrypts only the classification stage. HEAR~\cite{kim2022secure} runs CKKS-encrypted CNN inference on pre-extracted skeleton joints for fall detection, while RPRS~\cite{wu2025rprs} pairs an mmWave radar with an FHE server for trajectory recognition. Critically, RPRS encrypts only the CNN: the entire upstream DSP pipeline (range-FFT, Doppler-FFT, CFAR) executes in plaintext, leaving all intermediate representations exposed. NeuJeans~\cite{ju2024neujeans} accelerates homomorphic convolution. For encrypted signal processing, QASP~\cite{nguyen2025quantized} demonstrates FHE for audio STFT/MFCC. Visionary work by Aguilar-Melchor et al.~\cite{aguilar2013recent} and Lagendijk et al.~\cite{lagendijk2012encrypted} predicted FHE-based DSP but predated CKKS. Lastly, Mazzone et al.~\cite{mazzone2025efficient} \rev{proposed efficient} encrypted argmax \revdel{at 12.8--14.2\,s }on CPU.
% \name{} encrypts the \textit{entire} radar DSP pipeline---from raw range profiles through detection, phase extraction, and filtering---not just a downstream classifier. This requires reformulating non-linear, data-dependent operations as fixed polynomial circuits, and provides the stronger guarantee that no intermediate representation is ever exposed in plaintext.

\noindent\textbf{FHE Systems and Acceleration.}
Several open-source CKKS libraries underpin modern FHE applications: OpenFHE~\cite{al2022openfhe} (C++, bootstrapping, Chebyshev evaluation, Intel HEXL), Microsoft SEAL~\cite{peng2019danger} (simpler API, no bootstrapping), and TenSEAL~\cite{benaissa2021tenseal} (Python/PyTorch bindings over SEAL). Algorithmic advances further reduce overhead: minimax composite polynomials~\cite{cheon2022efficient} cut encrypted comparison cost by ${\sim}45\%$, and P2P-CKKS~\cite{osei2025p2p} optimizes encrypted FFT via dynamic padding. GPU-accelerated frameworks~\cite{agullo2025fideslib, ozcan2023homomorphic} achieve $100$--$2000\times$ speedups, bringing even bootstrapping into the sub-second regime on modern hardware. Taken together, these library, algorithmic, and hardware advances have shifted FHE from a theoretical concept to a deployable systems primitive, \rev{as evident in \name{}.} \revdel{\name{} builds directly on this maturing ecosystem: its entire encrypted radar pipeline runs end-to-end on commodity hardware today, and} Each generation of CKKS tooling only narrows the gap to real-time operation.

\section{Limitations and Future Work}
\label{sec:discussion}

\begin{revblock}
This paper shows the initial feasibility of executing an entire mmWave sensing pipeline under FHE. In this section, we discuss the current limitations of \name{} and future directions to address them.

\noindent\textbf{Sensing Performance.} \textit{Accuracy:}
The FHE-friendly approximations trade accuracy for data-oblivious fixed circuit designs. The gesture pipeline reaches 84.5\% versus RPRS~\cite{wu2025rprs}'s 95.0\%.
This gap comes from the simpler feature extraction
and the three-layer network that fit \name{}'s bootstrapping-free depth-11 budget, whereas RPRS computed richer feature maps in plaintext and fed a CNN. To close this gap, future work can adopt FHE inference of advanced architecture like transformer based models~\cite{chen2022x}.

\noindent\textit{Single Dominant Target:}
\name{}'s current pipeline is suited for a single target, as the soft power attention kernels collapse the scene to one dominant subject.
Future work can extend the \name{} library to multiple targets by utilizing Bartlett beamforming and higher-order power-moment kernels, combined with a public grid of range and angle sectors. This is feasible in FHE because each of the techniques above can be composed data-obliviously, and \name{} can run in parallel in every predefined sector in the grid.
\noindent The same fixed-circuit  approach extends to dense 3D heatmap construction and processing as well. Point cloud processing, however, primarily relies on data-dependent detection and clustering, which the \name{} kernel library currently does not support (Theorem \ref{thm:data_oblivious}). Future work will address designing kernels that produce fixed-size detection statistics that allow a secure client to recover point-cloud-like outputs after decryption.

\noindent\textit{Complex Environments:}
All three evaluation datasets record one line-of-sight subject at short range in ordinary indoor rooms; strong multipath, weak direct paths, and non-line-of-sight propagation are not explicitly tested in this work. Future work needs to validate \name{} in more complex and real-world challenging environments.

\noindent\textbf{Latency and Online Operation.} \name{} currently runs offline and in batch mode. Even though the batched packing brings the vital-signs pipeline below the sensor data rate (1.21\,s per 10\,s window), the gesture pipeline still takes 5.76\,s per 3\,s window.  Time-sensitive tasks such as fall detection are still over the latency budget of continuous real-time monitoring. To enable online operation in the future, first trivial requirement is an infrastructure that chains client encryption with cloud evaluation over sliding windows. Such streaming would also enable cross-frame state tracking. Second, we can reduce the FC layer costs through cyclic weight embedding and hoisted rotations \cite{halevi2014algorithms}. Lastly, the latency can be further improved by utilizing emerging FHE accelerators like~\cite{agullo2025fideslib, noauthorintelsnodate, noauthordprivenodate}, so that deeper models like transformers can run without bootstrapping.

\noindent\textbf{End-to-End Implementation.} \name{}'s current implementation encrypts range profiles, after lightweight clutter removal and normalization performed in plaintext on the trusted client. With enough compute, \name{} could easily be extended beyond range profiles to the actual raw ADC, because the K3 DFT kernel can easily be used in that case. The main challenge in either case is the dynamic range of the input data. In FHE, normalization is a critical step, and without it the encryption noise destroys the computation. Whatever the input (range profile or raw ADC), the client must estimate a stable normalization factor efficiently (preferably in linear time or better) before encryption to ensure that the data is in the stable numerical range of the underlying FHE scheme. In our experiments, a calibrated Parseval-based energy estimate showed promise in this regard.

\noindent\textbf{Malicious Cloud.} \name{} assumes a semi-honest cloud; a malicious cloud could evaluate a different circuit and return an unauthorized inference that the client would unknowingly decrypt. To address this, future work could explore verifiable FHE and zero-knowledge proofs of correct circuit evaluation, alongside lightweight client-side spot checks that recompute randomly sampled windows. A complementary direction is sensor attestation: a hardware root of trust on the radar device would sign captured range profiles, ensuring that encrypted inputs originate from a genuine sensor rather than fabricated data. Combining attested capture with proofs of computation can provide an end-to-end chain of custody, from the sensor through encrypted cloud evaluation to the decrypted result, towards truly private and cryptographically secure mmWave sensing.
\end{revblock}

%% ----------------------------------------------------------------
\section{Conclusion}
\label{sec:conclusion}

We presented \name{}, the first system to execute the entire mmWave radar DSP and ML inference pipeline on encrypted data. By replacing standard signal-processing routines with seven FHE-compatible kernels: energy integration, soft power attention, DFT, soft I/Q extraction of highest-energy bin, FIR filtering, notch masking, and differential phase extraction, \name{} enables \rev{a semi-honest} cloud to process \rev{encrypted} sensor streams without ever observing plaintext values. Evaluation on three public datasets (270 vital-sign recordings, 600 gesture trials) shows HR/RR estimation within ${<}\,10^{-3}$\,bpm of the plaintext baseline and 84.5\% gesture accuracy with 99.8\% prediction agreement. Formal analysis establishes input privacy and \rev{circuit-level} data obliviousness for any pipeline composed from the kernel library. \rev{With batched ciphertext packing, \name{} shows initial feasibility of end-to-end mmWave sensing under FHE on commodity hardware.} Thus, \name{} opens a path toward privacy-preserving wireless health monitoring and smart sensing applications.

\begin{acks}
This research was supported by the Digital Life Initiative Doctoral Fellowship at Cornell Tech.
\end{acks}

\bibliographystyle{ACM-Reference-Format}
\bibliography{ref}
\appendix
\section{FMCW Radar Sensing Primer}
\label{sec:app_primer}
\label{subsec:fmcw_primer}

\noindent\textbf{FMCW Sensing.}
A Frequency-Modulated Continuous Wave (FMCW) radar transmits a chirp waveform for sensing, a sinusoidal linear chirp of duration $T_c$ and slew rate $b$ is:
\begin{equation}
    s_t(t) = e^{j2\pi\left(f_c t + \frac{b t^2}{2}\right)}, \quad 0 \le t \le T_c
\end{equation}
where $f_c$ is the carrier frequency. If the signal reflects off a target at distance $d$ and is received as an attenuated, delayed echo:
\begin{equation}
    s_r(t) = \alpha \cdot s_t(t-\tau) \cdot e^{j2\pi f_d t}
\end{equation}
where $\alpha$ is the attenuation factor, $\tau = 2d/c_{\text{light}}$ is the round-trip delay ($c_{\text{light}}$ is the speed of light), and $f_d$ is the Doppler frequency.

\noindent\textbf{Mixing and IF Signal.}
The receiver multiplies the echo by the conjugate of the transmitted signal, producing an Intermediate Frequency (IF) signal:
\begin{equation}
    s_{IF}(t) \approx \mu \cdot \alpha \cdot e^{j(2\pi \Delta f \, t + \varphi(\tau))}
\end{equation}
where $\Delta f \approx 2bd/c_{\text{light}}$ is the beat frequency proportional to target range, and $\varphi(\tau) \approx 2\pi f_c \tau$ encodes fine-grained range information. The IF signal is digitized by an ADC at rate $f_s$, producing $M$ discrete samples $x[n] = I[n] + jQ[n]$ per chirp per antenna.

\noindent\textbf{Range FFT.}
Spectral estimation via an $M$-point FFT on the IF samples resolves targets along the range axis, based on $\Delta f$. For antenna $a \in [A]$ and chirp $c \in [D]$:
\begin{equation}
    \label{eq:range_fft}
    z_{a,c}[r] = \sum_{n=0}^{M-1} x_{a,c}[n] \, e^{-j \frac{2\pi}{M} r n}, \quad r \in [R]
\end{equation}
where $R$ is the number of range bins retained. The peak index in $|z_{a,c}[r]|$ localizes the target; the range resolution is $\Delta R = c_{\text{light}} / 2B$ with chirp bandwidth $B = bT_c$.

\noindent\textbf{Doppler FFT.}
Each frame consists of $D$ chirps transmitted consecutively. For a fixed range bin $r$ and antenna $a$, the sequence $\{z_{a,c}[r]\}_{c=0}^{D-1}$ across chirps encodes the Doppler (velocity) of targets at that range. A $D$-point FFT across the chirp dimension produces the range-Doppler spectrum:
\begin{equation}
    \label{eq:doppler_fft}
    Z_{a}[r,d] = \sum_{c=0}^{D-1} w[c] \cdot z_{a,c}[r] \, e^{-j\frac{2\pi}{D} d c}, \quad d \in [D]
\end{equation}
where $w[c]$ is a window function (e.g., Hanning) applied to reduce spectral leakage. Targets appear as peaks in the 2D range-Doppler map, with the Doppler index $d$ proportional to radial velocity.

\noindent\textbf{Phase Tracking.}
Vital signs (breathing and heart rate) manifest as millimeter-level chest displacements too small to shift the range-bin peak. Instead, the phase of the complex signal at the target range bin encodes these micro-motions. Small radial displacement $\Delta d$ induces a phase shift:
\begin{equation}
    \Delta \varphi = \frac{4\pi \Delta d}{\lambda}
\end{equation}
When a human subject is at $\hat{r}$ in the sensing environment, extracting $\angle z_{a,c}[\hat{r}]$ across frames yields a time series proportional to chest displacement, processed via bandpass filtering (0.1--0.5\,Hz for respiration, 0.8--2.0\,Hz for heart rate) and spectral estimation to recover HR and RR.

\noindent\textbf{Continuous Monitoring.}
The radar transmits $F$ frames, each containing $D$ chirps across $A$ antennas and $R$ range bins. A single frame yields a range-Doppler-antenna tensor $\mathbf{Z}[t] \in \mathbb{C}^{A \times R \times D}$; $t \in \{0, 1, 2, \ldots, F-1\}$. Over $F$ frames, each element $z_{a,c}[r,t]$ carries two channels:
\begin{itemize}
    \item \textbf{Magnitude} $|z|^2$: tracks macro-movement (target presence, range).
    \item \textbf{Phase} $\angle z$: encodes micro-movement (sub-wavelength displacements such as breathing and heartbeat).
\end{itemize}
This duality, that both channels are inseparably encoded in the same complex IQ samples; is the source of the biometric \rev{leakage} formalized in \S3.

\noindent\textbf{Standard Application Pipelines.}
Two canonical radar processing chains operate as follows:
\begin{enumerate}
    \item \textit{Vital signs:} range FFT $\rightarrow$ target detection (peak or CFAR) $\rightarrow$ phase extraction at target bin $\rightarrow$ bandpass filtering $\rightarrow$ PSD estimation $\rightarrow$ HR/RR.
    \item \textit{Gesture recognition:} range FFT $\rightarrow$ Doppler FFT $\rightarrow$ range-Doppler feature extraction $\rightarrow$ neural network classifier $\rightarrow$ gesture label.
\end{enumerate}
Both chains require spectral transforms, non-linear detection, and data-dependent operations that must be reformulated for encrypted computation (\S\ref{sec:system_design}).

\section{Proof of \revdel{Inseparability}\rev{Biometric Leakage} (Proposition~\ref{prop:biometric_inseparability})}
\label{sec:app_A}

% Removed the restated theorem (formerly thm:biometric_inseparability): it
% duplicated Proposition~\ref{prop:biometric_inseparability} with narrower
% wording, creating the numbering/naming mismatch reviewers flagged. The
% appendix now holds only the proof of the proposition stated in §2.

\begin{proof}
\noindent\textit{Algebraic inseparability:}
Energy-based detection requires computing $E_k[t] = |\mathbf{z}_k[t]|^2$ from plaintext range-bin values $\mathbf{z}_k[t] \in \mathbb{C}$. To evaluate this, the server \underline{\textit{must}} receive the complex samples $\mathbf{z}_k[t]$.
% (or equivalent raw ADC inputs $\mathbf{x}[t]$ from which $z_k[t]$ is recoverable via the invertible DFT: $\mathbf{z}[t] = \mathbf{F}\mathbf{x}[t]$). Each $z_k[t]$ decomposes uniquely in polar form as $(|z_k[t]|,\; \phi_k[t])$.
Providing $\mathbf{z}_k[t]$ for magnitude-based detection therefore simultaneously and unavoidably discloses the phase $\phi_k[t]$. A semi-honest server logs all inputs, making this leakage immediate and irreversible.

\noindent\textit{Phase--displacement coupling:}
By the FMCW radar equation (App. \ref{subsec:fmcw_primer}), small radial displacements $\Delta d_k[t]$ of the target at range bin $k$ induce phase shifts $\Delta \varphi_k[t] = {4\pi \Delta d_k[t]}/{\lambda}$. This mapping is linear and invertible, so the adversary recovers $\Delta d_k[t]$ from the disclosed phase sequence.

\noindent\textit{Nyquist sufficiency:}
The adversary now holds $F$ samples of $\Delta d_k[t]$ at rate $f_s \geq 2 f_{\mu}$ over a duration $F/f_s \geq T_{\min}$. By the Shannon--Nyquist theorem, this is sufficient to reconstruct the continuous-time micro-displacement signal and resolve its spectral components at all frequencies up to $f_{\mu}$.
\end{proof}

\section{Proof of Input Privacy (Theorem~\ref{thm:input_privacy})}
\label{sec:app_B}

\begin{proof}
By a standard hybrid argument over the $F$ frames. Suppose a PPT
adversary $\mathcal{A}$ distinguishes the two encrypted streams
with advantage~$\epsilon$. We construct a reduction $\mathcal{B}$
that, given an IND-CPA challenge ciphertext for a single frame,
embeds it at position~$t^{*}$ (encrypting
$\tilde{\mathbf{z}}_0[t]$ for $t < t^{*}$ and
$\tilde{\mathbf{z}}_1[t]$ for $t > t^{*}$) and forwards the
stream to $\mathcal{A}$. By the hybrid lemma, $\mathcal{B}$ breaks
IND-CPA with advantage $\geq \epsilon / F$, which must be
negligible. Since $F$ is polynomial in $\lambda$, $\epsilon \le
F \cdot \mathsf{negl}(\lambda)$ is itself negligible.

Combined with \revdel{Theorem}\rev{Proposition~\ref{prop:biometric_inseparability}}, this
shows that \name{} \revdel{breaks the inseparability barrier}\rev{closes
the biometric leakage channel}: the cloud
evaluates detection, spectral processing, phase extraction,
filtering, and classification circuits on
$\mathsf{Enc}(\tilde{\mathbf{z}}[t])$ but cannot recover any
plaintext value---neither the magnitude for unauthorized tracking
nor the phase for unauthorized vital-sign extraction.
\end{proof}

\section{Proof of Data Obliviousness (Theorem~\ref{thm:data_oblivious})}
\label{sec:app_C}

\begin{proof}
The proof proceeds in two steps.

\smallskip\noindent\textit{Step~1 (Per-kernel obliviousness).}
Each kernel $K_i$ defines a fixed arithmetic circuit $C_i$ over CKKS ciphertexts: the sequence of homomorphic operations (additions, multiplications, rotations) and their memory access pattern are determined entirely by the public configuration $(R,\allowbreak F,\allowbreak A,\allowbreak D,\allowbreak \gamma,\allowbreak \text{Taylor order},\allowbreak \text{filter order},\allowbreak \text{FC dims})$ and are independent of the encrypted input. We verify this for all seven kernels and the FC layer by construction in \S\ref{subsec:verify_per_kernel}.

\smallskip\noindent\textit{Step~2 (Composition closure).}
Let $C_1, C_2$ be data-oblivious circuits. The trace of their
sequential composition satisfies
$\mathrm{Trace}(C_2 \circ C_1,\, \mathbf{x})
 = \mathrm{Trace}(C_1,\, \mathbf{x})
 \;\|\; \mathrm{Trace}(C_2,\, C_1(\mathbf{x}))$.
Since $C_1$ is data-oblivious, the first component is identical
for all $\mathbf{x}$. Since $C_2$ is data-oblivious, the second
component is identical for all inputs, including the varying
intermediate ciphertexts $C_1(\mathbf{x})$. By induction over a
chain of $k$ circuits, any pipeline composed from this kernel set
has an input-independent trace.
\end{proof}

\subsection{Per-Kernel Verification}
\label{subsec:verify_per_kernel}

We verify that each kernel defines a fixed arithmetic circuit
whose trace depends only on public parameters.

\smallskip\noindent\textbf{K1: Energy Integration
(\S\ref{subsec:energy}).}
Circuit: one ct-ct squaring and one ct-ct addition per frame,
repeated $F$ times. Trace fixed by $(R,\allowbreak F)$.

\smallskip\noindent\textbf{K2: Soft Power Attention
(\S\ref{subsec:soft_argmax}).}
Circuit: $\log_2\gamma$ sequential squarings, one pt-ct multiply
(public index vector), and a rotation-based reduction. Trace fixed
by $(R,\allowbreak \gamma)$.

\smallskip\noindent\textbf{K3: DFT via Block-Diagonal Matmul
(\S\ref{subsec:doppler_dft}).}
Circuit: BSGS diagonal method on the fixed block-diagonal matrix
$\tilde{\mathbf{C}}$
(Eq.~\ref{eq:dft_re}). Trace fixed by
$(D,\allowbreak \text{window})$.

\smallskip\noindent\textbf{K4: Soft I/Q Extraction
(\S\ref{subsec:phase}).}
Circuit: computes $(|z|^2)^{P_\phi} \cdot z$ and reduces via
summation (Eq.~\ref{eq:phase_iq}). All $R$ bins processed
identically. Trace fixed by $(R,\allowbreak P_\phi)$.

\smallskip\noindent\textbf{K5: FIR Filtering
(\S\ref{subsec:fir}).}
Circuit: pt-ct multiply with public FIR coefficients
(Eq.~\ref{eq:fir_iq}). Backend (Toeplitz or rotation-based)
selected at configuration time. Trace fixed by
$(\text{filter order},\allowbreak F)$.

\smallskip\noindent\textbf{K6: Notch Mask
(\S\ref{subsec:notch}).}
Circuit: single pt-ct multiply with a fixed binary mask
(Eq.~\ref{eq:notch_mask}). Trace fixed by $D$.

\smallskip\noindent\textbf{K7: Differential Phase Extraction
(\S\ref{subsec:taylor_arctan}).}
Circuit: evaluates the fixed polynomial
$\Delta\phi[t] \approx y[t]\,x[t]^2 - \tfrac{1}{3}\,y[t]^3$
(Eq.~\ref{eq:taylor_arctan}). Trace fixed by
$(\text{Taylor order})$.

\smallskip\noindent\textbf{FC Inference
(\S\ref{subsec:mlp}).}
Circuit: $L$ pt-ct matmuls and $L{-}1$ ct-ct squarings
(Eq.~\ref{eq:mlp_forward}). The square activation $\sigma(x)=x^2$
has no conditional branches. Trace fixed by
$(d_0,\allowbreak d_1,\allowbreak \ldots,\allowbreak d_L)$.

\smallskip\noindent\textit{Remark.}
Per-frame amplitude normalization is performed client-side before
encryption; it does not affect the cloud-side trace.

\end{document}